\documentclass{article}
\usepackage[utf8]{inputenc}
\usepackage[english]{babel}
\usepackage{amsmath,amsthm,amssymb,amsfonts,bm}
\usepackage{graphicx,color,latexsym,multirow,url}
\usepackage{epsfig}
\usepackage{float}
\usepackage{lscape}
\usepackage{booktabs,tabularx,pbox}
\usepackage{adjustbox}
\usepackage{caption}
\usepackage{enumitem}
\usepackage{longtable}
\usepackage{mathtools}
\usepackage{authblk}
\usepackage{hyperref}
\hypersetup{
colorlinks=true,
linkcolor=blue,
citecolor=blue,
urlcolor=blue
}
\numberwithin{equation}{section}
\usepackage[round,sort&compress]{natbib}
\usepackage[position=top]{subfig}
\usepackage[boxed,lined,ruled]{algorithm2e}
\newenvironment{SMOGAlgorithm}[1][htb]
  {% Update algorithm name 
   \begin{algorithm}[#1]%
  }{\end{algorithm}}

\usepackage{mwe,tikz}
\usepackage{pgfplots}
\usepackage[percent]{overpic}
\usetikzlibrary{shapes,arrows,positioning,calc}
\tikzstyle{line} = [draw, -latex']

\newcommand{\norm}[1]{\lVert#1\rVert}
\definecolor{mygray}{RGB}{211,211,211}

\newtheorem{theorem}{Theorem}

\newtheorem{lemma}{Lemma}
\newtheorem{remark}{Remark}

\DeclareMathOperator*{\argmin}{arg\,min}

\usepackage{epstopdf}
\AppendGraphicsExtensions{.tiff}

\newcolumntype{C}[1]{>{\centering\arraybackslash}p{#1}} 

\title{Structural modeling using overlapped group penalties for discovering predictive biomarkers for subgroup analysis}
\author[1]{Chong Ma}
\author[1]{Wenxuan Deng}
\author[1]{Shuangge Ma}
\author[2]{Ray Liu}
\author[2]{Kevin Galinsky \thanks{Corresponding: Kevin.Galinsky@takeda.com}}

\affil[1]{Department of Biostatistics, School of Public Health, Yale University, New Haven, CT, 06510}
\affil[2]{Advanced Analytics and Statistical Consultation, Takeda Pharmaceuticals, Cambridge, MA 02139}

\date{}

\begin{document}

\maketitle

\begin{abstract}

The identification of predictive biomarkers from a large scale of covariates for subgroup analysis has attracted fundamental attention in medical research. In this article, we propose a generalized penalized regression method with a novel penalty function, for enforcing the hierarchy structure between the prognostic and predictive effects, such that a nonzero predictive effect must induce its ancestor prognostic effects being nonzero in the model. Our method is able to select useful predictive biomarkers by yielding a sparse, interpretable, and predictable model for subgroup analysis, and can deal with different types of response variable such as continuous, categorical, and time-to-event data. We show that our method is asymptotically consistent under some regularized conditions. To minimize the generalized penalized regression model, we propose a novel integrative optimization algorithm by integrating the majorization-minimization and the alternating direction method of multipliers, which is named after \texttt{smog}. The enriched simulation study and real case study demonstrate that our method is very powerful for discovering the true predictive biomarkers and identifying subgroups of patients.
% GLOG implements a hierarchical group lasso algorithm for the identification of predictive and prognostic biomarkers. In randomized control trials, GLOG can be used to select the predictive and prognostic effects of biomarkers while maintaining strong hierarchy, which dictates that if the predictive effect is included in the final model, then the prognostic effect must be included as well.\\

\noindent \sloppy 
{\textbf {Availability:}} The program is available as an R package and can be downloaded from 
% CRAN (\url{http://cran.r-project.org}) and 
GitHub (\href{https://github.com/chongma1989/smog}{https://github.com/chongma1989/smog}).\\

    
\end{abstract}

%------------------------------------------------------------%
% Section 1: Introduction
%------------------------------------------------------------%

\section{Introduction}\label{sec:Intro}
Subgroup analysis has attracted considerable attention in drug development for identifying patient subgroups that benefit from a greater treatment effect in various diseases \citep{lipkovich2017tutorial}.
Patient subgroups can be determined by individual baseline characteristics such as demographic, clinical, and genomic covariates, which we refer to as biomarkers.
These biomarkers can be categorized as \textit{predictive} or \textit{prognostic}.
Prognostic biomarkers are defined as those that predict the trajectory of the disease \textit{in the absence of treatment},
while predictive biomarkers are defined as those that predict the trajectory of the disease \textit{in response to a particular treatment} \citep{ruberg2015personalized}.

Consider, for example, the head and neck squamous cell carcinoma (HNSCC) data from The Cancer Genome ATLAS (TCGA) study. which contain the overall survival time in months and the whole genome expression profiles of the tumor sample for each of the 510 HNSCC patients. The patients are dichotomized by whether or not have received radiation therapy, in which 301 patients are with radiation treatment and 209 without. Among the 510 HNSCC patients, 211 are alive and 299 deceased. Note that the gene expressions were obtained from the tumor samples of the patients who did not have neo-adjuvant therapy at the beginning of the cohort study. Figure~\ref{fig:hnsc_surv0} displays the Kaplan-Meier curves for the whole HNSCC data, and the biomarker positive subgroup of patients identified by using the proposed method in this article, respectively. Note that the treatment effect for the whole HNSCC data is statistically significant with p-value $=0.005$, though it might not be clinically significant. Therefore, it is intriguing to investigate if there exist predictive genomic biomarkers that can be used to identify a subgroup of patients who benefit more treatment effect. 

\begin{figure}[hp]
    \centering
    \includegraphics[width=0.5\textwidth,height=2in]{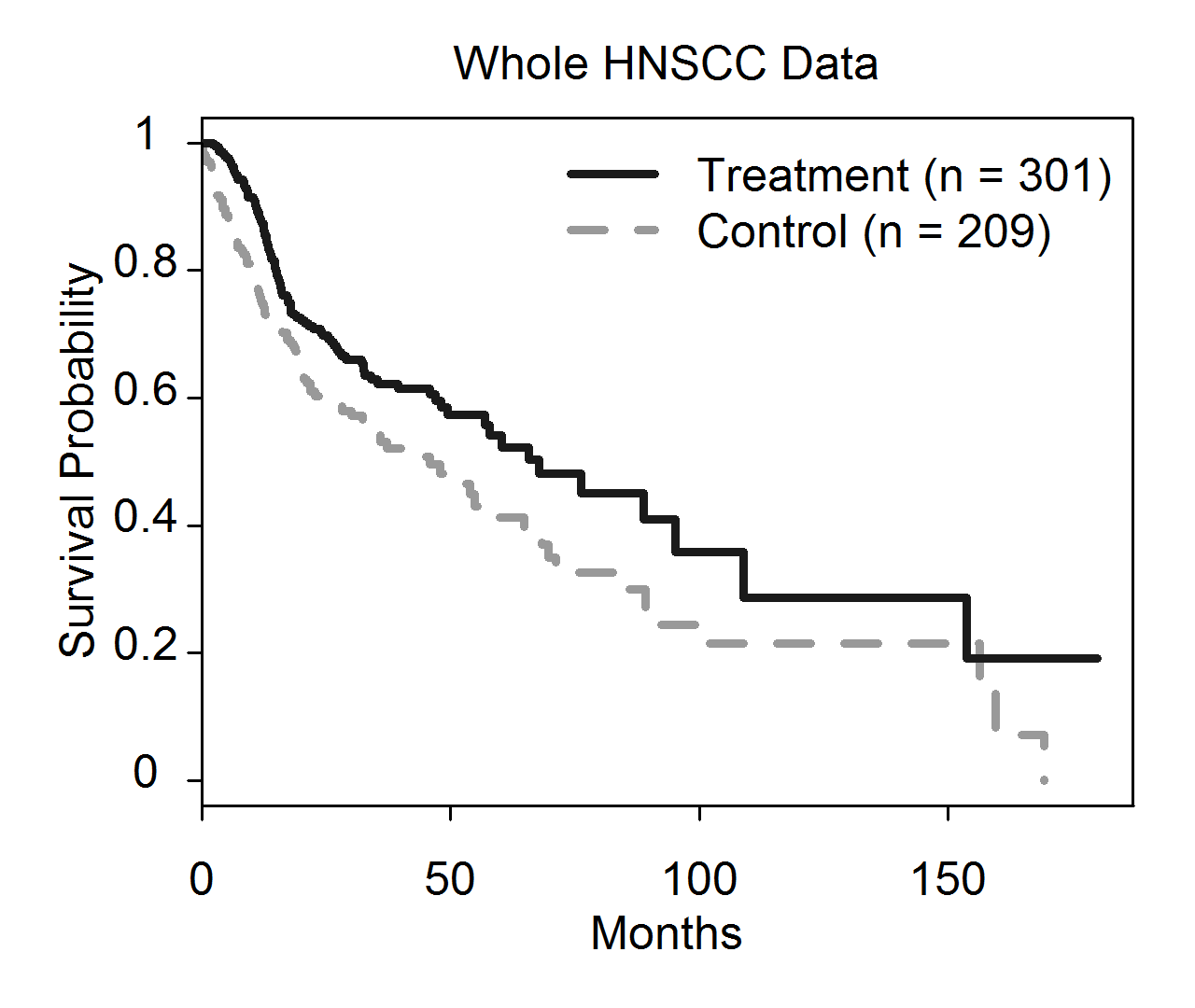}%
    \includegraphics[width=0.5\textwidth,height=2in]{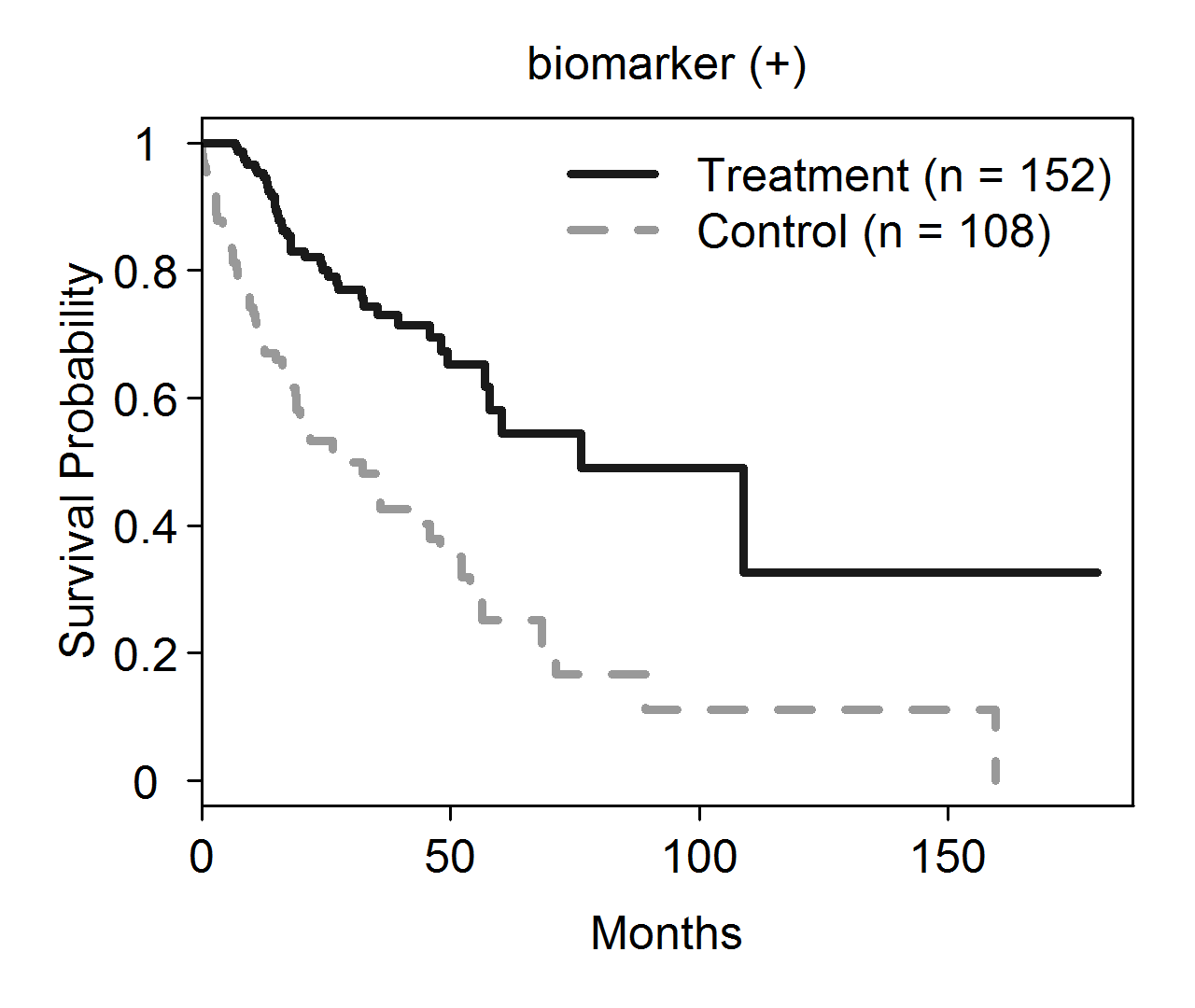}\\
    \caption{Kaplan-Meier curves for the whole HNSCC data and the positive biomarker subgroup (biomarker (+)) identified by using the proposed method, respectively. The 95\% confidence intervals of relative risks for the radiation treatment versus no radiation treatment are $(0.52,0.89)$ and $(0.23,0.51)$, accordingly.}
    \label{fig:hnsc_surv0}
\end{figure}

% some literatue review about overlapped group lasso
% The subgroup analysis has attracked considerable attention in the clinical trial community. \citet{lipkovich2017tutorial} overviewed some principled subgroup analysis approaches including the global outcome modeling methods and tree-based regression models \cite{loh2015regression}. 

The subgroup analysis and biomarker investigation mainly underlie the variable selection methods by selecting some potential variables from a large scale of candidate covariates and their interactions with treatment. Penalized regression methods have been utilized substantially for biomarker study by modeling the outcome as some link function of the prognostic effects of baseline covariates and their predictive effects, while imposing on specific penalty function of the coefficients of the predictor variables for different goals \citep{imai2013estimating,zhao2013effectively}. Note that the prognostic effects refer to as the main effects of the baseline covariates, and the predictive effects as the interaction effects of the baseline covariates with the treatment indicator. 

In the general topic of the penalized regression methods, there exists considerably extensive literature for solving various pragmatic and theoretic problems \citep{tibshirani1996regression,yuan2006model,meinshausen2009lasso,huang2008adaptive,liu2012incorporating}. Hierarchical structure models have been studied tremendously for honoring peculiar hierarchy structures between the prognostic and predictive effects by proposing various penalty functions \citep{zhao2009composite,bien2013lasso,jenatton2010proximal}. \citet{radchenko2010variable} proposed the VANISH method, which is designed for high-dimensional nonlinear problems by using a penalized least square criterion to enforce the heredity constraint, that is, if an interaction term is included in the model, then the according main effects are included automatically. Although VANISH is proved to equip sparsistency and persistence properties under some regularity conditions, in practice it has little chance to select the interaction terms. 
% In general, the overlapped Lasso method involves a composition of $L_p$ norms of the coefficients for the groups of predictors for $p \ge 1$, hence the penalty term $\Omega(\bm{\beta},\bm{\gamma})$ in \eqref{fun1} is not separable completely. Therefore, we can not apply the popular block coordinate descent algorithm \citep{tseng2001convergence} to minimize such objective function. For the similar sake, \cite{bien2013lasso} proposed to add a set of convex constraints to the lasso for selecting pairwise interaction predictors, by satisfying the hierarchy structure such that an interaction must appear in the model if only if one or both its ancestors are included in the model. \cite{jenatton2010proximal} proposed to apply a specified sequence of elementary proximal operators for a sparse hierarchical structure model, which is proved to be linear in $\mathcal{O}(p)$ when the penalty is in $L_1$, $L_2$, and $L_{\infty}$ norms.  

\cite{lim2015learning} proposed an overlapped group-Lasso approach to select pairwise interactions. Their method augmented and combined the main effects and the corresponding interaction effects as new groups, conditioning such constraints that the strong hierarchy structure between the main and interaction effects are satisfied, and then applied the group-Lasso approach \citep{yuan2006model} for model fitting by using the fast iterative shrinkage-thresholding algorithm (FISTA) \citep{beck2009fast}. The addition of the original and augmented main effects are the ultimate estimated main effects in the model. \citet{lim2015learning} developed a R pacakge \texttt{glinternet}, which is used to compete with our developed package \texttt{smog}. 

In this paper, we propose a novel penalized regression method for identifying predictive genomic biomarkers by using a specific penalty function of the prognostic and predictive effects to enforce the specified hierarchy structure~\eqref{hierstr}, that is, if a gene's predictive effect is nonzero, its corresponding prognostic effect must be nonzero as well. Note that the treatment effect is not included in the penalty function.
\begin{align}\label{hierstr}
    \text{predictive effect} \ne 0 \Longrightarrow \text{prognostic effect} \ne 0
\end{align}

In order to deal with various types of response variables such as continuous, categorical, and survival data, the proposed penalized regression method is adaptively built under the Guaussian, logistic, and Cox proportional hazard models, respectively. Because the prognostic and predictive effects are overlapped in the penalty function, the popular block coordinate descent algorithm \citep{tseng2001convergence} can not be readily applied to minimize the penalized regression model. To address this issue, we propose a novel optimization algorithm to integrate the majorization-minimization (MM) \citep{figueiredo2007majorization} and the alternating direction method of multipliers (ADMM) \citep{tseng1991applications} to solve the minimization problem in the penalized regression method. Our algorithm demonstrates more powerful performance in the simulation study and the real case study. 

This paper is organized as follows. We propose our method and algorithm in Section~\ref{sec:Method}. Section~\ref{sec:tuning} introduces several model selection criteria and proposes a greedy path search for the optimal penalty parameters. We study the asymptotic property for the proposed method under Gaussian assumption in Section~\ref{sec:asymp}. In Section~\ref{sec:Simulation}, we conduct various simulation studies to compare our approach with some existing methods for detecting such genetic biomarkers that satisfies the hierarchy structure \eqref{hierstr}. We explore our method for some real clinical trial data in Section~\ref{sec:Data analysis}, and discuss some extended insights and future work in Section~\ref{sec:Discuss}.

%------------------------------------------------------------%
% Section 2: methodology
%------------------------------------------------------------%
\section{Methodology} \label{sec:Method}

\subsection{Notation} \label{sec:Method-Notation}
For convenience, we summarize the notations that would be used in the following sections. Denote that $f(\tau,\bm{\beta},\bm{\gamma}|\text{data})$ is the loss function, that is composed by the negative log-likelihood function for the data under specified model assumptions. The loss function relates the response variable to the addition of linear terms of the predictors under certain link function, denoted by $\text{link}(\bm{y}) \sim \tau \bm{t} +  \bm{X}\bm{\beta} + \bm{W}\bm{\gamma}$, where $\bm{t}$ denotes the treatment indicators, $\bm{X}=(\bm{x}_{(1)},\ldots,\bm{x}_{(d)})$ denotes the matrix of $d$ predictors (genes), and $\bm{W}=(\bm{x}_{(1)} \times \bm{t},\ldots,\bm{x}_{(d)} \times \bm{t})$ represents the matrix of $d$ predictor-by-treatment interactions. Here note that $\bm{x}_{(j)} \times \bm{t}$ is the element-wise product between the $j$th predictor and the treatment $\bm{t}$. We denote that $\tau$ is the treatment effect, $\bm{\beta}=(\beta_1,\ldots,\beta_d)'$ are $d$ prognostic effects, and $\bm{\gamma}=(\gamma_1,\ldots,\gamma_d)'$ are $d$ predictive effects, respectively. Denote that $\tilde{\bm{X}}=(\bm{t},\bm{X},\bm{W})$ is the model matrix, and $\tilde{\bm{x}}_i$ is the $i$th row in $\tilde{\bm{X}}$. $\bm{\theta}' = (\tau,\bm{\beta}',\bm{\gamma}')$ represents all parameters to be estimated, $\psi_i = \tau t_i + \bm{\beta}'\bm{x}_i + \bm{\gamma}'\bm{w}_i$  represents the addition of linear terms of the predictors for the $i$th observation, and $\phi_i = e^{\psi_i}$ denotes the exponential of the linear effects $\psi_i$, respectively. Without loss of generality, we do not include the intercept in the model, since the intercept would be neglected when fitting the model for centralized data. Thus, the overall dimension of the predictors would be $p=2d+1$.  

%-------------------------------------------------------------%
\subsection{Penalized regression models} \label{sec:Method-model}
In this article, the objective function to minimize for generalized penalized models is 
\begin{equation}\label{fun1}
    L(\bm{\tau},\bm{\beta},\lambda|\text{data}) = f(\bm{\tau},\bm{\beta},\bm{\gamma}|\text{data}) + \lambda \Omega(\bm{\beta},\bm{\gamma}),
\end{equation}
Where $f(\bm{\tau},\bm{\beta},\bm{\gamma}|\text{data})$ is the aforementioned loss function for the data, and $\lambda \Omega(\bm{\beta},\bm{\gamma})$ is the penalty function
\begin{equation} \label{penalty}
    \bm{\lambda} \Omega(\bm{\beta},\bm{\gamma}) = \sum_{j=1}^{d} \{ \lambda_1 \norm{(\beta_j,\gamma_j)}_2 + \lambda_2 \norm{(\beta_j,\gamma_j)}_2^2 + \lambda_3 |\gamma_j|  \},
\end{equation}
Which is a combination of the composition of the group Lasso penalty, ridge penalty, and Lasso penalty for each group of prognostic ($\beta_j$) and predictive ($\gamma_j$) effects. Lemma~\ref{lemma1} and lemma~\ref{lemma2} prove that the proposed penalty function enforces the specified hierarchy structure \eqref{hierstr}.

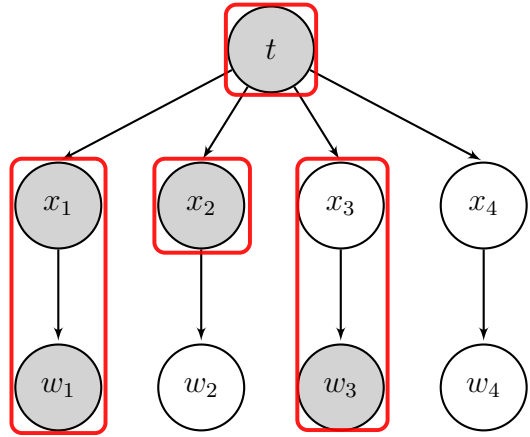
\begin{figure}[hp]
\centering

\begin{tikzpicture}[->,>=stealth',shorten >=1pt,
                    node distance=1.5cm,
                    thick,auto,
                    main node/.style=
                    {circle,draw,
                    inner sep=0.5pt, 
                    minimum size=3.25em,
                    font=\fontsize{12}{12}\selectfont
                    }]

  \node[main node] (1) [fill=mygray]{$t$};
  \node[main node] (2) [below left=1.25cm and 2cm of 1, fill=mygray] {$x_1$};
  \node[main node] (21) [below = 1.25cm of 2,fill=mygray] {$w_1$};
  
  \node[main node] (3) [below left=1.25cm and 0.1cm of 1,fill=mygray] {$x_2$};
   \node[main node] (31) [below = 1.25cm of 3] {$w_2$};
  
%   \node[main node] (j) [below = 0.85cm of 1] {\ldots \ldots};
%      \node[main node] (j1) [below = 1.25cm of j] {\ldots \ldots};
     
  \node[main node] (4) [below right=1.25cm and 0.1cm of 1] {$x_{3}$};
   \node[main node] (41) [below = 1.25cm of 4,fill=mygray] {$w_{3}$};
  
  \node[main node] (5) [below right=1.25cm and 2cm of 1] {$x_4$};
   \node[main node] (51) [below = 1.25cm of 5] {$w_4$};
  
  %draw paths 
  \path [line] (1) -- (2.north);
  \path [line] (2) -- (21);
  
  \path [line] (1) -- (3.north);
  \path [line] (3) -- (31);
  
%   \path [line] (1) -- (j);
%   \path [line] (j) -- (j1);
    
  \path [line] (1) -- (4.north);
  \path [line] (4) -- (41);
  
  \path [line] (1) -- (5.north);
  \path [line] (5) -- (51);
\end{tikzpicture}

\begin{tikzpicture}[remember picture, overlay]
% top circle
\draw [red,opacity=0.9, line width=1.5pt, rounded corners] (-0.6,4.9) rectangle (0.6,6.1);

% circle for x1
\draw [red,opacity=0.9, line width=1.5pt, rounded corners] (-3.45,0.4) rectangle (-2.2,4.05);

% circle for x2
\draw [red,opacity=0.9, line width=1.5pt, rounded corners] (-1.55,2.8) rectangle (-0.28,4.05);

% circle for x3
\draw [red,opacity=0.9, line width=1.5pt, rounded corners] (0.35,0.44) rectangle (1.55,4.05);

\end{tikzpicture}

\caption{Hierarchical structure for genetic biomarkers. A large predictive effect $\gamma_1$ would protect its corresponding prognostic effect $\beta_1$ in the model.}
\label{fig: hier1}
\end{figure}

%-------------------------------------------------------------%
% \subsection{Penalty function enforcing the hierarchy structure} \label{sec:Method-pen}
\cite{zhao2009composite} proposed the Composite Absolute Penalties (CAP) families for selecting groups of predictors by satisfying the provided known hierarchy structures at the across-group and within-group levels, and suggested using the BLASSO algorithm to search the CAP regularization path. Hence, our proposed penalty function can be seen as a specific case of the CAP family. 

% In order to honor the hierarchy structure~\eqref{hierstr}, we propose a new penalty function in the CAP family,
% \begin{equation} \label{penalty}
%     \bm{\lambda} \Omega(\bm{\beta},\bm{\gamma}) = \sum_{j=1}^{d} \{ \lambda_1 \norm{(\beta_j,\gamma_j)}_2 + \lambda_2 \norm{(\beta_j,\gamma_j)}_2^2 + \lambda_3 |\gamma_j|  \}
% \end{equation}
% Which is a combination of the composition of the group Lasso penalty, ridge penalty, and Lasso penalty for each group of prognostic ($\beta_j$) and predictive ($\gamma_j$) effects. 

Nonetheless, the design of the penalty function~\eqref{penalty} serves exclusively for detecting predictive biomarkers in drug development, which is compromised among the model complexity, interpretability, and predictability, respectively. Notice that the penalty function $\Omega(\bm{\beta},\bm{\gamma})$ merely includes the prognostic and predictive effects of covariates, precluding the penalization of the treatment effect $\tau$, because the treatment effect should not be penalized in medical research. We do not include higher order interactions when building the model for controlling the model complexity. The overlapping of the predictive effect $\gamma_j$ between the group Lasso penalty $\norm{(\beta_j,\gamma_j)}_2$ and the Lasso penalty $|\gamma_j|$ works for honoring the hierarchy structure~\eqref{hierstr}, which is demonstrated in Lemma~\ref{lemma1} and \ref{lemma2}. The ridge penalty aims to improve the consistency of the estimation for the groups of the prognostic and predictive effects $\beta_j$ and $\gamma_j$, when the multicollinearity occurs \citep{zou2005regularization}. We recommend keep the ridge penalty parameter $\lambda_2$ as a fixed tiny value, such as $\varepsilon = 10^{-6}$ \citep{bien2013lasso}, and merely tune the group Lasso penalty parameter $\lambda_1$ and the Lasso penalty parameter $\lambda_3$. We give more details in the greedy pathway search for $\lambda_1$ and $\lambda_3$ in Section~\ref{sec:tuning}.

Figure~\ref{fig: hier1} illustrates the ideal case in which the biomarkers would be selected by the penalized regression method. When the predictive effect $\gamma_2$ (the coefficient for $x_2 \times t$) is small, a moderately large $\lambda_3$ would shrink $\gamma_2$ to be exactly zero, while its corresponding prognostic effect $\beta_2$ (the coefficient for $x_2$) would be kept in the model, once $\lambda_1$ is not too large. On the other hand, if a predictive effect $\gamma_1$ would not be shrunken to be exactly zero, the group Lasso penalty would keep both of $\gamma_1$ and its corresponding prognostic effect $\beta_1$ in the model. In 

% Section~\ref{sec:Method-overlap}, we present more details in Lemma~\ref{lemma1} and \ref{lemma2}. 

% In fact, our specified hierarchy structure is a weak hierarchy, since the nonzero predictive effect does not force the treatment effect being nonzero. Nonetheless, because the treatment effect is not penalized so that it would be always included in the model, our specified hierarchy structure is also a strong hierarchy. 

We introduce the penalty function used in this article to enforce the hierarchy structure~\eqref{hierstr}, and next we briefly discuss the loss function for various types of response variable such as continuous, categorical, and time-to-event data. In this article, we assume that the continuous response variable underlies the Guassian assumption, the categorical response variable underlies the multinomial logistic regression model, and the time-to-event response variable underlies the Cox proportional hazard model, respectively. 

% \subsubsection{Linear regression model under the Gaussian assumption} \label{sec:Method-model-gaussian}

\begin{enumerate}[label={(\arabic*)}]
    \sloppy
    \item \textbf{Continuous response variable.} Assume $y_i \overset{i.i.d}{\sim} N(\psi_i,\sigma^2)$. Thus, $f(\tau,\bm{\beta},\bm{\gamma}|\mathrm{data})$ in~\eqref{fun1} is the common squared-error loss function $\frac{1}{2}\sum_{i=1}^{n} \left( y_i - \psi_i \right)^2 $, that is, 
\begin{equation} \label{model}
f(\bm{\tau},\bm{B},\bm{\Gamma}|\text{data}) = \frac{1}{2}\norm{\bm{y} - (\tau\bm{t} + \sum_{j=1}^{d} ( \beta_j \bm{x}_{(j)} + \gamma_j \bm{x}_{(j)} \times \bm{t}) )}^2 
\end{equation}
% \noindent
% Where $\tau$ is the treatment effect, and $(\beta_j,\gamma_j)$ denote the prognostic and predictive effects for the $j$th predictor, respectively. Note that $\bm{x}_{(j)}$ represents the $j$th predictor, and $\bm{x}_{(j)} \times \bm{t}$ is the element-wise product between the $j$th predictor and the treatment $\bm{t}$. The penalty term is an addition of composite of $L_2$ and $L_1$ norms on the prognostic and predictive effects for each predictor, which is proved to satisfy the hierarchy structure in section~\ref{sec:Method-model}. 

    \item \textbf{Categorical response variable.}
    Suppose that there are $K$ categories for the response variables. Use the outcome $K$ as the pivot outcome, then the multinomial regression model under some regularized assumptions gives us that $\mathrm{logit}(P(Y_i = k)) = \tau_k t_i + \bm{\beta}_k' \bm{x}_i + \bm{\gamma}_k' \bm{w}_i$, $i=1,\ldots,n; k = 1,\ldots,K-1$. Hence, the loss function underlying the multinomial logistic regression model is the negative log-likelihood function 
\begin{equation}
    f(\bm{\tau},\bm{B},\bm{\Gamma}|\text{data}) = 
    \sum_{i=1}^{n}\mathrm{log}\left( 1 + \sum_{k=1}^{K-1} \phi_{i,k} \right) 
    - \sum_{k=1}^{K-1} \sum_{i \in \mathcal{G}_k} \psi_{i,k}
\end{equation}
Where $\psi_{i,k} = \tau_k t_{i} + \bm{\beta}_k'\bm{x}_i + \bm{\gamma}_k'\bm{w}_i$, $\phi_{i,k} = e^{\psi_{i,k}}$, and $\bm{\theta}_k' = (\tau_k, \bm{\beta}_k', \bm{\gamma}_k')$ is the set of regression coefficients associated with outcome $k$. Note that $\bm{\tau} = (\tau_1,\ldots,\tau_{K-1})$, $\bm{B}=(\bm{\beta}_1,\ldots,\bm{\beta}_{K-1})$, and $\bm{\Gamma}=(\bm{\gamma}_1,\ldots,\bm{\gamma}_{K-1})$ represent the combined treatment effects, prognostic effects, and predictive effects, respectively. 

    \item \textbf{Time-to-event response variable.}
    Let $T_i$ and $C_i$ be the failure time and censoring time for subject $i$, $i=1,\ldots,n$. Define that $Y_i = \min\{T_i,C_i\}$, and $\delta_i=I(T_i \le C_i)$ is the censoring indicator. $(Y_i,\delta_i,\tilde{\bm{x}}_i),i=1,\ldots,n$ is the observed data, and we assume that $T_i$ follows the Cox proportional hazard model. Note that the hazard function is $h(Y_i|\tilde{\bm{x}}_i) = h_0(Y_i)\phi_i$. For simplicity, we use the negative partial log-likelihood function as the loss function to estimate the effect of the predictors, for avoiding modeling the baseline hazard function $h_0$. The negative partial log-likelihood function is 
\begin{equation}
\small
    f(\tau,\bm{\beta},\bm{\gamma}|\text{data}) = 
    \sum_{\delta_i=1} \left( \sum_{l=0}^{m_i-1} \mathrm{log} \left( \sum_{j:Y_j \ge t_i} \phi_j 
    - \frac{l}{m_i} \sum_{j \in \mathrm{H}_i} \phi_j \right) - 
    \sum_{j \in \mathrm{H}_i} \psi_j \right) 
\end{equation}
Where $t_i$ denotes the unique noncensored times, $\mathrm{H}_i = \{Y_j: Y_j = t_i, \delta_j=1 \}$ denotes the set of tied noncensored times, and $m_i$ is the size of $\mathrm{H}_i$, respectively.
\end{enumerate}

% \subsubsection{Multinomial logistic regression model} \label{sec:Method-model-logistic}
% \subsubsection{Cox proportional hazard model} \label{sec:Method-model-Coxph}

So far, we present the penalized regression models for three different types of response variables. Then we introduce a novel optimization algorithm to minimize the penalized regression model, and explain why the penalty function $\Omega(\bm{\beta},\bm{\gamma})$ enforces the hierarchy structure~\eqref{hierstr}. 

% comment: the kronecker product might not be appropriate here. 
% e.g., T is n X 1, X is n X p, and T \otimes X is n^2 X p, 
% which breaches the formula. 
% \begin{equation}
% \begin{aligned}
%     Y & = T \tau + X \beta + \left[ \diag(T) \times X \right] \gamma \\
%     \hat{\theta} & = \argmin \left[
%         \left( Y - f \left(\theta; T, X\right) \right)^2 +
%         \lambda_1 \sum_i \left< \beta_i, \gamma_i \right> +
%         \lambda_2 \sum_i \left< \gamma_i \right> +
%         \lambda_3 \sum_i \left< \beta_i, \gamma_i \right>^2
%     \right]
% \end{aligned}
% \end{equation}

\subsection{Integrative Optimization Algorithm} \label{sec:Method-algorithm}
% In this subsection, we present a novel optimization algorithm to minimize the penalized linear function $L(\bm{\theta}|\text{data}) = f(\bm{\theta}|\text{data}) + \bm{\lambda} \Omega(\bm{\theta})$, where $f(\bm{\theta}|\text{data})$ is the loss function, and $\bm{\lambda} \Omega(\bm{\theta}) = \sum_{j=1}^{d} \big( \lambda_1 \norm{(\beta_j,\gamma_j)}_2 + \lambda_2 \norm{(\beta_j,\gamma_j)}_2^2 + \lambda_3 |\gamma_j|  \big) $ is the penalty function, which will be proved to satisfy the hierarchy structure~\eqref{hierstr}, that is, $\gamma_j \ne 0 \Longrightarrow \beta_j \ne 0$.

Notice that $\beta_j$ and $\gamma_j$ are a group of related prognostic and predictive effects for the predictor $x_j$, and $\gamma_j$ is overlapped in the group Lasso penalty $\norm{(\beta_j,\gamma_j)}_2$ and the Lasso penalty $|\gamma_j|$, therefore the penalty function $\Omega(\bm{\theta})$ is block separable between groups of prognostic and predictive effects, rather than component separable within each group of them. Hence, the minimization problem of the penalized regression model can not be readily solved by using the block coordinate descent algorithm. 

\sloppy
To address this issue, we propose to a novel optimization algorithm by integrating the Majorization-Minimization (MM) and the alternative direction method of multipliers (ADMM) to minimize the penalized regression model $L(\bm{\theta}) = f(\bm{\theta}) + \bm{\lambda} \Omega(\bm{\theta})$, where $f(\bm{\theta}|\text{data})$ is the loss function, and $\bm{\lambda} \Omega(\bm{\theta}) = \sum_{j=1}^{d} \big( \lambda_1 \norm{(\beta_j,\gamma_j)}_2 + \lambda_2 \norm{(\beta_j,\gamma_j)}_2^2 + \lambda_3 |\gamma_j|  \big) $ is the penalty function, which will be proved to enforce the hierarchy structure~\eqref{hierstr}, that is, $\gamma_j \ne 0 \Longrightarrow \beta_j \ne 0$. In ADMM form, the minimization problem for $L(\bm{\theta})$ is equivalent to 
\begin{equation} \label{Lagrangian}
    \begin{split}
         \argmin L(\bm{\theta}) & = f(\bm{\theta}) + \bm{\lambda} \Omega(\bm{z}) \\
    \text{subject to } & \bm{\theta} = \bm{z}
    \end{split}
\end{equation}
Which can be written as the augmented Lagrangian function
\begin{equation} \label{augLag}
\tilde{L}(\bm{\theta},\bm{z},\bm{u},\rho) =  f(\bm{\theta}) + \bm{\lambda} \Omega(\bm{z}) + \bm{u}'(\bm{\theta} - \bm{z}) + \frac{\rho}{2}\norm{\bm{\theta} - \bm{z}}^2,    
\end{equation}
where $\bm{z}$ is the augmented parameter that is equivalent to $\bm{\theta}$, $\bm{u}$ is the dual variable, and $\rho>0$ is the penalty for the deviation of $\bm{\theta}$ and $\bm{z}$. Given the warm start for the parameters $\bm{\theta}^k$, $\bm{z}^k$, $\bm{u}^k$, the minimization of the augmented Lagrangian function~\eqref{augLag} boils down to alternating the minimization with respect to $\bm{\theta}$, $\bm{z}$, and the dual variable $\bm{u}$, respectively, while holding other variables constant, until the stopping criterion is satisfied such that the primal residual and the dual residual converge to zero as the algorithm proceeds. 

Without loss of generality, assume that $f(\bm{\theta})$ is a continuously differentiable convex function with Lipschitz contant $L$ of $\nabla f(\bm{\theta})$. $f(\bm{\theta})$ can be approximated by its majorization function using the second-order Taylor expansion at $\bm{\theta}^k$, that is, $
f(\bm{\theta}) \le Q(\bm{\theta}|\bm{\theta}^k) = f(\bm{\theta}^k) + \nabla f(\bm{\theta}^k)'(\bm{\theta} - \bm{\theta}^k) + \frac{L}{2} \norm{\bm{\theta} - \bm{\theta}^k}^2$. Then, the $(k+1)$th $\bm{\theta}$-minimization turns out  
$\bm{\theta}^{k+1} = \argmin Q(\bm{\theta}|\bm{\theta}^k) + \frac{\rho}{2} \norm{\bm{\theta} - (\bm{z}^k - \bm{u}^k)}^2$. Hence, $\bm{\theta}^{k+1} = \frac{1}{L+\rho} \left( L\bm{\theta}^k - \nabla f(\bm{\theta}^k) + \rho (\bm{z}^k - \bm{u}^k)\right)$. Note that when $f(\bm{\theta})$ is nonlinear, it is not trivial to calculate the smallest Lipschitz constant of $\nabla f(\bm{\theta})$ ($L$). We use the backtrack line search approach to obtain the smallest Lipschitz constant for $L$. To be brief, one can initialize some small value for $L$ and some positive value $\alpha>1$, repeat updating $L$ by scaling it by $\alpha$ until $f(\bm{\theta}^{k+1}) > Q(\bm{\theta}^{k+1} | \bm{\theta}^k)$. This novel optimization algorithm is named for \texttt{smog}, which is designed to minimize structural models using overlapped group penalty. The overarching steps are summarized in Algorithm~\ref{algorithm:MM-ADMM0}. 

\begin{SMOGAlgorithm}[H]
    \SetAlgoLined
    \caption{MM-ADMM for penalized regression models}
    \label{algorithm:MM-ADMM0}
    \KwIn{function $f,\Omega$, and parameters $\bm{\lambda},\rho$}
    Initialization for $\bm{\theta}^0$, $\bm{z}^0$, $\bm{u}^0$, $\alpha > 1$\;
    \While{$\varepsilon^{\mathrm{pri}} > e^{\mathrm{pri}}$ or $\varepsilon^{\mathrm{dual}} > e^{\mathrm{dual}}$}{
    % Majorization-Minimization
    Initialization for some $L>0$\;
    \Repeat{$f(\bm{\theta}) \le Q(\bm{\theta}|\bm{\theta}^k)$}{
    $Q(\bm{\theta}|\bm{\theta}^k) = f(\bm{\theta}^k) + \nabla f(\bm{\theta}^k)'(\bm{\theta} - \bm{\theta}^k) + \frac{L}{2} \norm{\bm{\theta} - \bm{\theta}^k}^2$ \;
    Update: $L = \alpha L$ \quad (Backtrack line search for $L$) \;
    }
    % ADMM 
    $\bm{\theta}^{k+1} = \argmin Q(\bm{\theta}|\bm{\theta}^k) + \frac{\rho}{2} \norm{\bm{\theta} - (\bm{z}^k - \bm{u}^k)}^2$\;
    $\bm{z}^{k+1} = \argmin \frac{\rho}{2} \norm{\bm{z} - (\bm{\theta}^{k+1} + \bm{u}^k)}^2 + \bm{\lambda} \Omega(\bm{z})$\;
    $\bm{u}^{k+1} = \bm{u}^{k} + \bm{\theta}^{k+1} - \bm{z}^{k+1}$\;
    }
\end{SMOGAlgorithm}

% \subsection{Overlapped group-Lasso for the hierarchy structure}\label{sec:Method-overlap}
Notice that the update of $\bm{\theta}^{k+1}$ is trivial, since the corresponding objective function $Q(\bm{\theta}|\bm{\theta}^{k})+\frac{L}{2} \norm{\bm{\theta} - \bm{\theta}^k}^2$ is in quadratic form. The crucial is to update $\bm{z}^{k+1}$, therefore we present the Lemma~\ref{lemma1} and \ref{lemma2} to show that the minimizer $\bm{z}^{k+1}$ is fundamentally the proximal operator of the penalty function $\Omega$ with parameter $\lambda/\rho$ on $\bm{\theta}^{k+1}+\bm{u}^k$, that is, $\bm{z}^{k+1}=\mathrm{prox}_{\rho\Omega/\bm{\lambda}}(\bm{\theta}^{k+1}+\bm{u}^k)$, where $\bm{u}^k$ is essentially the running sum of the residuals between the original primal variable $\bm{\theta}$ and the augmented variable $\bm{z}$. 
\begin{lemma} \label{lemma1}
Let $\bm{\theta}=(\theta_1,\theta_2)$, $\bm{b}=(b_1,b_2)$. The solution to minimize  $l(\bm{\theta}) = \frac{1}{2}\norm{\bm{\theta} - \bm{b}}^2 + \lambda_1 \norm{\bm{\theta}}_2 + \lambda_2 \norm{\bm{\theta}}^2 + \lambda_3 |\theta_2|$ with respect to $\bm{\theta}$ is in such form that 
\begin{equation}\label{lemma1:solution}
    \begin{cases}
        \theta_1 = \frac{1}{1+2\lambda_2} (1 - \frac{\lambda_1}{\sqrt{|b_1|^2 + |s(b_2,\lambda_3)|^2}} )_{+} b_1 \\
        \theta_2 = \frac{1}{1+2\lambda_2} (1 - \frac{\lambda_1}{\sqrt{|b_1|^2 + |s(b_2,\lambda_3)|^2}} )_{+} s(b_2,\lambda_3)\\
    \end{cases}
\end{equation}
In general, $s(b_2,\lambda_3)=\text{sign}(b_2)(|b_2|-\lambda_3)_{+}$ is the Lasso operator on $b_2$ for the soft-thresholding parameter $\lambda_3$, where $t_{+}=\max(t,0)$.  
\end{lemma}

\begin{proof}
Take the derivative for $l(\bm{\theta})$ with respect to $\theta_1$ and $\theta_2$ respectively, and when $\bm{\theta} \ne 0$, we have 
\begin{equation}\label{lemma1:derivative}
\begin{cases}
\theta_1(1+2\lambda_2 + \frac{\lambda_1}{\norm{\bm{\theta}}_2}) = b_1\\
\theta_2(1+2\lambda_2 + \frac{\lambda_1}{\norm{\bm{\theta}}_2}) = b_2 - \lambda_3 \text{sign}(\theta_2)\\
\end{cases}
\end{equation}
After taking some simple algebra for~\eqref{lemma1:derivative}, we obtain $\norm{\bm{\theta}} = \frac{1}{1+2\lambda_2} \left (\sqrt{|b_1|^2+|s(b_2,\lambda_3)|^2} - \lambda_1 \right)_{+}$; When $\sqrt{|b_1|^2+|s(b_2,\lambda_3)|^2} \le \lambda_1$, $\bm{\theta} = 0$. Substitute it back in the formula~\eqref{lemma1:derivative}, and we can get the solution~\eqref{lemma1:solution}. 
\end{proof}

% mention something about KKT

\begin{remark}\label{remark1}
In Lemma~\ref{lemma1}, $\theta_1$ and $\theta_2$ represent the prognostic and predictive effects, and the penalty term $\bm{\lambda}\Omega(\bm{\theta}) = \lambda_1 \norm{\bm{\theta}}_2 + \lambda_2 \norm{\bm{\theta}}^2 + \lambda_3 |\theta_2|$ realizes the overlapped penalty on the predictive effect $\theta_2$. The three penalty parameters play their specific roles in this constrained minimization, where $\lambda_1$ controls to select or discard $\theta_1$ and $\theta_2$ together, $\lambda_2$ can alleviate the multicollinearity between $\theta_1$ and $\theta_2$, and $\lambda_3$ controls the selection of $\theta_2$. 

We can tell from the solution~\eqref{lemma1:solution} that the hierarchy structure $\theta_2 \ne 0 \Longrightarrow \theta_1 \ne 0$. Because $(\theta_1,\theta_2)$ is essentially obtained by applying the group Lasso operator on $(b_1,s(b_2,\lambda_3))$, $\theta_1$ and $\theta_2$ are selected or discarded simultaneously. When $s(b_2,\lambda_3) \ne 0$, $\theta_1 \ne 0$ must come up with $\theta_2 \ne 0$ for $\sqrt{|b_1|^2+|s(b_2,\lambda_3)|^2} > \lambda_1$; when $s(b_2,\lambda_3) = 0$, $\theta_2$ must be zero but $\theta_1 = \frac{1}{1+2\lambda_2}s(b_1, \lambda_1)$, which depends on $\lambda_1$.   
\end{remark}

\begin{lemma} \label{lemma2}
Let $\bm{\theta}=(\tau,\bm{\beta},\bm{\gamma})$, $\bm{b}=(b_0,\bm{b}_1,\bm{b}_2)$, where $\tau,b_0 \in \mathbb{R}$, and $\bm{\beta},\bm{\gamma}$ and $\bm{b}_1,\bm{b_2}$ $\in \mathbb{R}^d$. Consider the minimization problem 
\begin{equation}\label{lemma2:min}
  \min_{\bm{\theta} \in \mathbb{R}^{2d+1}} \frac{1}{2}\norm{\bm{\theta} - \bm{b}}^2 + 
  \sum_{j=1}^{d} \big( \lambda_1 \norm{(\beta_j,\gamma_j)}_2 + \lambda_2 \norm{(\beta_j,\gamma_j)}_2^2 + \lambda_3 |\gamma_j|  \big)
\end{equation}
Where $\lambda_1,\lambda_2,\lambda_3$ are both positive. The solution to the regularized minimization problem~\eqref{lemma2:min} is $\bm{\theta}=\mathrm{prox}_{\Omega/\bm{\lambda}}(\bm{b})$ such that 
\begin{equation}\label{lemma2:solution}
    \begin{cases}
    \tau = b_0 \\
    \beta_j = \frac{1}{1+2\lambda_2}(1-\frac{\lambda_1}{\sqrt{|b_1^j|^2+|s(b_2^j,\lambda_3)|^2}})_{+}b_1^j \\
    \gamma_j = \frac{1}{1+2\lambda_2}(1-\frac{\lambda_1}{\sqrt{|b_1^j|^2+|s(b_2^j,\lambda_3)|^2}})_{+}s(b_2^j,\lambda_3)
    \end{cases}
\end{equation}
$j=1,\ldots,d$. Note that $s(x,\lambda)=\text{sign}(x)(|x|-\lambda)_{+}$ is the lasso soft-threshold operator for $\lambda > 0$. 
\end{lemma}

\begin{proof}
Because the objective function~\eqref{lemma2:min} are separate for $\tau$, $(\beta_j,\gamma_j)$, $j=1,\ldots,d$, thus it can be solved by the block coordinates independently. Note that the penalty term in the minimization problem~\eqref{lemma2:min} does not contain $\tau$, therefore $\tau = b_0$. By applying the Lemma~\ref{lemma1}, we can directly obtain the solution~\eqref{lemma2:solution}. 
\end{proof}

\begin{remark}\label{remark3}
Lemma~\ref{lemma2} is used for $\bm{z}$-update in \texttt{smog}, which realizes the soft-thresholding by satisfying the hierarchical structure within the groups of $(\beta_j,\gamma_j)$. In fact, we merely apply the Majorization-Minimization algorithm to update $\bm{\theta}$ one time, but we do not intend to reach the global minimum value for a temporary objective function $f(\bm{\theta}) + \frac{\rho}{2} \norm{\bm{\theta} - (\bm{z}^k - \bm{u}^k)}^2$ for updating $\bm{\theta}$, since such temporary global minimum value might be trapped in some local minimum of ~\eqref{augLag}.      
Given the provided absolute and relative tolerance values $\varepsilon^{\mathrm{abs}}$ and $\varepsilon^{\mathrm{rel}}$ (usually set as $10^{-4}$ or $10^{-5}$), the stopping criterion in \texttt{smog} requires that the primal residual $\varepsilon^{\mathrm{pri}}=\frac{\norm{\bm{\theta}^{k+1} - \bm{z}^{k+1}}_2}{\sqrt{2p}}$, and the dual residual $\varepsilon^{\mathrm{dual}}=\frac{\norm{\bm{z}^{k+1} - \bm{z}^{k}}_2}{\sqrt{2p}}$ satisfy that 
\begin{align} \label{stopcriterion}
 \varepsilon^{\mathrm{pri}} & \le e^{\mathrm{pri}} = \varepsilon^{\text{abs}} + \varepsilon^{\mathrm{rel}} \max\{\frac{\norm{\bm{\theta}}_2^{k+1}}{\sqrt{2p}},\frac{\norm{\bm{z}}_2^{k+1}}{\sqrt{2p}}\} \\ 
 \varepsilon^{\mathrm{dual}} & \le e^{\mathrm{dual}} = \sqrt{\frac{n}{p}}\rho^{-1}\varepsilon^{\mathrm{abs}} + \varepsilon^{\mathrm{rel}} \frac{\norm{\bm{u}}_2^{k+1}}{\sqrt{2p}}
\end{align}
\end{remark}

% \begin{algorithm}[H]
% \SetAlgoLined
% \KwIn{$\bm{\tilde{X}},\bm{y}$, and parameters $\bm{\lambda},\rho$}
% Initialization for $\bm{\theta}^0$, $\bm{z}^0$, $\bm{u}^0$, $\alpha > 1$\;
% \While{$\varepsilon^{\mathrm{pri}} > e^{\mathrm{pri}}$ or $\varepsilon^{\mathrm{dual}} > e^{\mathrm{dual}}$}{

% $\bm{\theta}^{k+1} 
% =(\bm{\tilde{X}}^T\bm{\tilde{X}} + \rho \bm{I}_{n})^{-1}(\bm{\tilde{X}}^T\bm{y}+\rho
% (\bm{z}^{k}-u^k))$ \;
% \For{$j \gets 1$ to $d$}{
% $\bm{z}_{j}^{k+1} = \text{prox}_{\bm{\lambda}/\rho}(\bm{\theta}_{j}^{k+1}+\bm{u}_{j}^{k})$
% \quad $\left(\bm{\theta}_j^{k+1} = (\beta_j^{k+1},\gamma_j^{k+1}) \right)$\;
% }
% $\bm{u}^{k+1}=\bm{u}^{k}+\bm{\theta}^{k+1}-\bm{Z}^{k+1}$\;
% }
% \caption{MM-ADMM for continuous outcome}
% \label{algorithm:MM-ADMM1}
% \end{algorithm}

\subsection{Tuning the penalty parameter $\mathbf{\lambda}$}\label{sec:tuning}
% \subsection{Model selection criteria}

We considered several approaches to determine the optimal penalty parameters $\bm{\lambda}$. The mean square errors of the response from $K$-fold cross-validations is a popular approach to select the best predictive model. One can randomly divide the data into $K$ folds, where the $K-1$ folds of data are used to fit a model, which is in turn to validate the left one fold data by using the mean square errors of the response. This approach can be described by using 
\[
\mathrm{MSE}(\bm{\lambda})=\frac{1}{K}\sum_{k=1}^{K}\norm{\bm{y}_{k} - \hat{\bm{y}}_{(-k)}^{\bm{\lambda}}}^2,
\]
where $\hat{\bm{y}}_{(-k)}^{\bm{\lambda}}$ is the fitted response for the $k$th fold data without using themselves. A similar approach is the generalized cross-validation score, which is defined as 
\[
\mathrm{GCV}(\bm{\lambda}) = \frac{\norm{\bm{y} - \hat{\bm{y}}_{\bm{\lambda}}}^2}{n(1-d_{\bm{\lambda}}/n)^2}.
\]
Note that $d_{\bm{\lambda}}$ is the degrees of freedom that is approximated by $\mathrm{tr}\left\{\tilde{X}_{\bm{\lambda}}(\tilde{X}_{\bm{\lambda}}^T\tilde{X}_{\bm{\lambda}}+W_{\bm{\lambda}})^{-1}\tilde{X}_{\bm{\lambda}^T}\right\}$, where $\tilde{X}_{\bm{\lambda}}$ is the submatrix with columns corresponding to the nonzero coefficients selected by $\bm{\lambda}$, and $W_{\bm{\lambda}}$ is the diagonal matrix with the elements that are also determined by the nonzero coefficients. Because $\tau$ is not included in the penalty function, $W_{\bm{\lambda}}(\hat{\tau})=0$. And $W_{\bm{\lambda}}(\hat{\beta}_j)= \frac{\lambda_1}{\norm{(\hat{\beta}_j,\hat{\gamma}_j)}_2} + 2\lambda_2$ for $|\hat{\beta}_j| \ne 0$, and $W_{\bm{\lambda}}(\hat{\gamma}_j)= \frac{\lambda_1}{\norm{(\hat{\beta}_j,\hat{\gamma}_j)}_2} + 2\lambda_2 + \frac{\lambda_3}{|\hat{\gamma}_j|}$ for $|\hat{\gamma}_j| \ne 0$, respectively. We also considered the AIC and BIC type criterion to search the optional $\bm{\lambda}$, such that 
$$\mathrm{AIC}(\bm{\lambda}) = \mathrm{log}(\norm{\bm{y}-\hat{\bm{y}}_{\bm{\lambda}}}^2/n)+2d_{\bm{\lambda}}/n,$$
and 
$$\mathrm{BIC}(\bm{\lambda}) = \mathrm{log}(\norm{\bm{y}-\hat{\bm{y}}_{\bm{\lambda}}}^2/n)+\mathrm{log}(n)d_{\bm{\lambda}}/n.$$
In practice, the sample size $n$ is usually small in genetic biomarker study, an correction AIC criterion is borrowed from \cite{zhao2009composite} that is $$\mathrm{cAIC}(\bm{\lambda}) = \frac{n}{2}\mathrm{log}(\norm{\bm{y}-\hat{\bm{y}}_{\bm{\lambda}}}^2) + 
\frac{n}{2}(1+d_{\bm{\lambda}}/n)/(1-d_{\bm{\lambda}}+2/n),$$
in which $d_{\bm{\lambda}}$ is used as the number of nonzero predictors. 

For simplicity notation, we define $c(\bm{\lambda})$ by the model selection criterion including the $\mathrm{MSE}(\bm{\lambda})$, $\mathrm{GCV}(\bm{\lambda})$, $\mathrm{AIC}(\bm{\lambda})$, $\mathrm{BIC}(\bm{\lambda})$, and $\mathrm{cAIC}(\bm{\lambda})$, respectively. In Section~\ref{sec:Simulation}, our simulation study demonstrates that the generalized cross-validation (GCV) and AIC criteria are likely to select more covariates, which results in both higher sensitivity and false discovery rate in terms of the true covariates in the model, and cAIC and BIC are inclined to select smaller number of covariates, causing the total opposite situation as GCV and AIC. MSE seems more balanced than other four model assessments, though it is more computationally expensive. Overall, it is tricky to claim one criterion overwhelmingly dominates other criteria, therefore we recommend one could consider use both $\mathrm{GCV}$ and $\mathrm{MSE}$ criteria under specific conditions.  

% \subsection{Greedy pathway search for the optimal $\mathbf{\lambda}$}
In practice, it is computationally unfeasible to search a fine grid of points for three tuning parameters $\lambda_1,\lambda_2,\lambda_3$ in the elastic-net type penalty term in~\eqref{penalty}. In Section~\ref{sec:Method-model}, we mentioned that $\lambda_2$ is recommend to be kept as a fixed tiny value, such as $10^{-6}$, because our simulation study finds that it is preferred by very small values. In Section~\ref{sec:Simulation}, for simplicity, we conduct the simulation study by keeping the ridge penalty parameter $\lambda_2$ as zero, and the greedy pathway search merely focuses on tuning $\lambda_1$ and $\lambda_3$. Such modification illustrates easily the power of selecting the true predictive genetic biomarkers. We make it optional to set $\lambda_2$ as any constant tiny value in the R package \texttt{smog}.

\cite{jacob2009group} proposed to combine bracketing and golden section search for univariate $\lambda$, however, it is still computationally expensive to do such searching for two $\lambda$'s. We propose a greedy pathway search for $\lambda_1$ and $\lambda_3$ by entitling the superiority search in $\lambda_1$ than $\lambda_3$, since $\lambda_1$ controls the selection for groups of prognostic and predictive effects. Because our proposed algorithm could yield very sparse estimates for large $\lambda$'s fast, it is wisdom to downward search the optimal $\lambda_1$ and $\lambda_3$ from a large value $\lambda_0$. Hence, start with a large enough tuning parameter, $\lambda_0=c/\lambda_{min}(\tilde{\bm{X}})$, in which $c>1$ ensures that $\lambda_0$ allows no more than one predictor entering the model. At the $k$th search step, define $\lambda_{1,k} = \lambda_{3,k} = \lambda_0\delta^{k}$, and calculate $c_1=c(\lambda_0\delta^{k+1},\lambda_0\delta_{k})$, $c_2=c(\lambda_0\delta^{k},\lambda_0\delta^{k+1})$, and $c_3=c(\lambda_0\delta^{k+1},\lambda_0\delta^{k+1})$, where $\delta$ is a fixed arbitrary damping ratio $\in (0,1)$. At the $(k+1)$th step, $\lambda_{1,k+1}=\lambda_0\delta^{k+I(c_1 \wedge c_3 \le c_2)}$ and $\lambda_{3,k+1}=\lambda_0\delta^{k+I(c_2 \wedge c_3 < c_1)}$. Continue this pathway search until $c(\lambda_{1,k+1},\lambda_{3,k+1}) > c(\lambda_{1,k},\lambda_{3,k})$ first occurs.

% \begin{figure}[hp]
%     \centering
%     \begin{tikzpicture}
%     \draw[step=1cm, gray, thin] (-1,-1) grid (5,5);
%     \draw[thick, ->, >=stealth'] (-1,-1) -- node[label = below: $\lambda_1$]{} (5.5,-1);
%     \draw[thick, ->, >=stealth'] (-1,-1) -- node[label = left: $\lambda_3$]{} (-1,5.5);
    
%     % draw a pathway search for lambda
%     \draw[thick, ->, >=stealth'] (5,5) -- (4,5);
%     \draw[thick, ->, >=stealth'] (4,5) -- (4,4);
%     \draw[thick, ->, >=stealth'] (4,4) -- (3,3);
%     \draw[thick, ->, >=stealth'] (3,3) -- (3,2);
%     \draw[thick, ->, >=stealth'] (3,2) -- (2,2);
%     \node[right] at (5,5) {(1,1)};
%     \node[below] at (-1,-1) {0};
%     \node[below] at (5,-1) {1};
%     \node[left] at (-1,5) {1};
    
%     \end{tikzpicture}
    
%     % draw a contour plot at center (2,2)
%     \begin{tikzpicture}[remember picture, overlay]
%     \begin{axis}[view={0}{90}, at={(-0.24\linewidth,0.08\linewidth)},
%                  hide axis, xmin=-3, xmax=3, ymin=-3, ymax=4]
%     \addplot3[contour gnuplot={number=10,labels=false,draw color=red}]
%     {exp(-1/2*((x+0.5)^2+(y-1)^2-1.2*(x+0.5)*(y-1)))};
%     \end{axis}
%     \end{tikzpicture}
%     \caption{Illustrate the greedy pathway search for $\lambda_1$ and $\lambda_3$. \label{fig:pathway}}
% \end{figure}

This greedy pathway search could yield a local minimum for a model selection criterion at a certain pair of $\lambda_1$ and $\lambda_3$, however, it might slip over the local minimum such that there does not exist a ``V" type overturning point for $\lambda_1$ and $\lambda_3$. Hence, we recommend set up the maximum search step as 20, and determine the ultimate $\lambda_1$ and $\lambda_3$ with the minimum model selection criterion value. On the other hand, the pathway search could be substantively affected by the damping ratio $\delta$, therefore we recommend one could try several damping ratios $\delta$'s and choose the pair of $\lambda_1$ and $\lambda_3$ with the minimum $c(\lambda_1,\lambda_3)$. We use $\delta=0.9$ in the simulation section.  

%============================================================%
% section: Asymptotic properties
%============================================================%
\section{Asymptotic properties}\label{sec:asymp}
In this section, we study the asymptotic property for the penalized regression model under the Gaussian assumption by using the analogous approach as \cite{knight2000asymptotics}. Under the Guassian assumption, the penalized regression function is 
\begin{equation}
\small
    \frac{1}{2}\norm{\bm{y} - \big(\tau\bm{t} + \sum_{j=1}^{d} ( \beta_j \bm{x}_{(j)} + \gamma_j \bm{x}_{(j)} \times \bm{t} ) \big)}^2  + \sum_{j=1}^{d} \big( \lambda_1 \norm{(\beta_j,\gamma_j)}_2 + \lambda_2 \norm{(\beta_j,\gamma_j)}_2^2 + \lambda_3 |\gamma_j|  \big)
\end{equation}
Before heading the Theorem~\ref{them1}, we denote some additional notations for convenience. Denote by $\breve{\bm{X}}$ the augmented design matrix for absorbing the ridge penalty term into the objective function, such that 
\[
\breve{\bm{X}} = 
\begin{bmatrix}
\bm{t} & \bm{X} & \bm{W} \\
\bm{0} & \sqrt{\lambda_2}\bm{I} & \bm{0} \\
\bm{0} & \bm{0} & \sqrt{\lambda_2}\bm{I} \\
\end{bmatrix} 
\]
Denote that $\bm{u}=(u_1,\bm{u}_2^T,\bm{u}_3^T)^T$ where $u_1$ is a vector of $1 \times 1$, and $\bm{u}_2$ and $\bm{u}_3$ are vectors of $d \times 1$. 

\begin{theorem}\label{them1}
If $\lambda_{l,n}/\surd{n} \to \lambda_{l} \ge 0 (l=1,2,3)$, and $\bm{C} = \lim_{n \to \infty} \left( \frac{1}{n} \breve{\bm{X}}^T \breve{\bm{X}}\right)$ is nonsingular, then 
\[
\surd{n}(\hat{\bm{\theta}}_n - \bm{\theta}) \underset{d}{\to} \argmin(V),
\]
where 
\begin{align*}
    V(\bm{u}) = &-\bm{u}^T\bm{B} + \bm{u^TCu} + 
    \lambda_1 \sum_{j=1}^{d} \{ \norm{(u_{2j},u_{3j})}_2 I(\norm{(\beta_j,\gamma_j)}_2 = 0) \\
    & + \frac{1}{\norm{(\beta_j,\gamma_j)}}_2 (u_{2j}|\beta_j| + u_{3j}|\gamma_j|)I(\norm{(\beta_j,\gamma_j)}_2 \ne 0) \}\\
    & + \lambda_3 \sum_{j=1}^{d} \{|u_{3j}|I(\gamma_j = 0)+u_{3j}sign(\gamma_j)I(\gamma_j \ne 0)\}
\end{align*}
and $\bm{B}$ has an $\mathcal{N}(\bm{0},\sigma^2C)$
\end{theorem}

\begin{proof}
Define $V_n(\bm{u})$ by 
\begin{align*}
V_n(\bm{u}) = & \frac{1}{2} \left \{\norm{\bm{\varepsilon} - \breve{\bm{X}}\bm{u}/\surd{n}}_2^2 - \norm{\varepsilon}_2^2 \right \}\\
 & + \lambda_{1,n} \sum_{j=1}^{d} \{ \norm{(\beta_j+u_{2j}/\surd{n},\gamma_j+u_{3j}/\surd{n})}_2 - \norm{(\beta_j,\gamma_j}_2 \} \\
 & + \lambda_{3,n} \sum_{j=1}^{d} \{ |\gamma_j + u_{3j}/\surd{n}| - |\gamma_j| \}
\end{align*}
Note that $V_n$ is minimized at $\surd{n}(\hat{\bm{\theta}}_n - \bm{\theta})$, and 
\[
\frac{1}{2} \left \{ \norm{\bm{\varepsilon} - \breve{\bm{X}}\bm{u}/\surd{n}}_2^2 - \norm{\varepsilon}_2^2 \right \} \to_d -\bm{u}^T\bm{B} + \bm{u^TCu}
\]
with finite-dimensional convergence holding trivially. We also have 
\begin{align*}
    & \lambda_{1,n} \sum_{j=1}^{d} \{ \norm{(\beta_j+u_{2j}/\surd{n},\gamma_j+u_{3j}/\surd{n})}_2 - \norm{(\beta_j,\gamma_j}_2 \\
    & \to
    \lambda_1 \sum_{j=1}^{d} ( \norm{(u_{2j},u_{3j})}_2 I(\norm{(\beta_j,\gamma_j)}_2 = 0) \\
    & \quad + \frac{1}{\norm{(\beta_j,\gamma_j)}_2} (u_{2j}|\beta_j| + u_{3j}|\gamma_j|)I(\norm{(\beta_j,\gamma_j)} \ne 0) )
\end{align*}
and 
\begin{align*}
\lambda_{3,n} \sum_{j=1}^{d} \{ |\gamma_j + u_{3j}/\surd{n}| - |\gamma_j| \} 
 \to \lambda_3 \sum_{j=1}^{d} (|u_{3j}|I(\gamma_j = 0)+u_{3j}sign(\gamma_j)I(\gamma_j \ne 0))
\end{align*}
Hence, $V_n(\bm{u}) \to_d V(\bm{u})$ (as defined above) with the finite-dimensional convergence holding trivially. Because $V_n$ is convex and $V$ has a unique minimum, it follows \citep{geyer1994asymptotics} that 
\[
\argmin(V_n) = \surd{n}(\hat{\bm{\theta}}_n - \bm{\theta}) \to_d \argmin(V)
\]
Note that when $\lambda_l=0$, $\argmin(V)=\bm{C^{-1}B} \sim \mathcal{N}(0,\sigma^2\bm{C}^{-1})$. 
\end{proof}
\begin{remark}\label{asym_remark}
In the proof, we assume that $d$ is fixed with $n \to \infty$, though it would be prefered to have $d=d_n \to \infty$. For the logistic and Cox proportional hazard models, the penalized linear regression model is more complex such that we do not prove the asymptotic property in this article. Under some regularized conditions, however, we conjecture that the asymptotic property could be extended to the generalized penalized regression model under the exponential family \citep{zou2006adaptive}. 
\end{remark}

\section{Simulation Study}\label{sec:Simulation}
In this section, we conduct extensive simulation studies to evaluate the performance of our proposed method in four scenarios by comparing our proposed method with Lasso, group Lasso, and overlapped group Lasso methods using the public-known R packages \texttt{glmnet}, \texttt{gglasso}, and \texttt{glinternet}, respectively. 

The R package \texttt{glmnet} was established \citep{friedman2010regularization,simon2011regularization,tibshirani2012strong} for implementing the Lasso method for various types of response variable, including the continuous, categorical, time-to-event, and count data, etc. under different model assumptions, by using coordinate descent algorithm. \cite{yang2015fast} proposed a groupwise majorization-descent algorithm for building the software \texttt{gglasso} to solve the group-Lasso learning problem. \texttt{glinternet} was developed by \cite{lim2015learning} for learning pairwise interactions in a linear or logistic regression model by enforcing the strong hierarchy: a nonzero interaction effect must induce its ancestor main effects being nonzero. 

\begin{table}[h]
    \centering
     \caption{Four scenarios used in the simulation study.}
    \label{tab:model_desc}
    \begin{adjustbox}{width=\textwidth}
    \begin{tabular}{cl}
    \hline \hline 
    \multicolumn{1}{c}{Scenario} & 
    \multicolumn{1}{c}{Description} \\
    \hline 
     \multirow{3}{*}{I}
     & Weak hierarchy: Only includes 5 biomarkers with prognostic effect. \\
     & $\beta_j = 0.2$, $\gamma_j = 0$, $j=1,\ldots,5$. \\
     & $\beta_j = 0$, $\gamma_j = 0$, $j=6,\ldots,d$. \\
     \hline 
     \multirow{3}{*}{II}
     & Anti-hierarchy: Only includes 5 biomarkers with predictive effect.\\ 
     & $\beta_j = 0$, $\gamma_j = 0.2$, $j=1,\ldots,5$. \\
     & $\beta_j = 0$, $\gamma_j = 0$, $j=6,\ldots,d$. \\
     \hline 
     \multirow{4}{*}{III}
     & Strong hierarchy: Includes 5 biomarkers with both prognostic \\
     & and predictive effects.\\
     & $\beta_j = 0.2$, $\gamma_j = 0.2$, $j=1,\ldots,5$. \\
     & $\beta_j = 0$, $\gamma_j = 0$, $j=6,\ldots,d$. \\
     \hline 
     \multirow{5}{*}{IV}
     & Mixture hierarchy: A combination of Model I, II, and III \\
     & $\beta_j = 0.14$, $\gamma_j = 0$, $j=1,\ldots,5$. \\
     & $\beta_j = 0$, $\gamma_j = 0.14$, $j=6,\ldots,10$. \\
     & $\beta_j = 0.14$, $\gamma_j = 0.14$, $j=11,\ldots,15$. \\
     & $\beta_j = 0$, $\gamma_j = 0$, $j=16,\ldots,d$. \\
     \hline \hline 
    \end{tabular}
    \end{adjustbox}
\end{table}

In fact, merely can \texttt{glinternet} compare with our software \texttt{smog} fairly, because these two software both aim to solve the penalized regression model with overlapped group penalty for enforcing the hierarchy structure~\eqref{hierstr}, though they are developed by using different algorithms. \texttt{glmnet} and \texttt{gglasso} would not always honor the hierarchy structure between prognostic and predictive effects, so that their fitted models lack interpretability. However, it still merits to study the performance of these softwares on different situations. 

% Note that merely \texttt{glinternet} does satisfy the hierarchy structure between the prognostic and predictive effects in~\eqref{hierstr}, but neither \texttt{glmnet} nor \texttt{gglasso} honors~\eqref{hierstr}. Though it is not fair to compare our method with \texttt{glmnet} and \texttt{gglasso}, it is still meaningful to verify how different the power is between these methods when ignoring the interpretability of models.  

\subsection{simulation scenarios}\label{subsec:sim_scen}
Table~\ref{tab:model_desc} displays the four scenarios used in this simulation study. The four scenarios are designed based on different hierarchy structures between prognostic and predictive effects in the model \citep{bien2013lasso}, which are weak hierarchy, anti-hierarchy, strong hierarchy, and mixed hierarchy, respectively. The weak hierarchy refers to as $\beta_j \ne 0,\gamma_j = 0$, the anti-hierarchy is $\beta_j = 0, \gamma_j \ne 0$, the strong hierarchy is $\beta_j \ne 0, \gamma_j \ne 0$, and the mixed hierarchy contains the three aforementioned hierarchy structures simultaneously in the model.

In each scenario, the simulation study is conducted underlying different sample sizes and number of biomarkers, such that $n=100,200$ and $d=200,500,1000$. Scenario I merely includes 5 biomarkers with prognostic effects but no predictive effects, scenario II includes 5 biomarkers with predictive effects but no prognostic effects, scenario III includes 5 markers with both prognostic and predictive effects, and scenario IV is a combination of scenarios I, II, and III as in Table~\ref{tab:model_desc}. 

% \begin{figure}[hp]
%     \centering
%     \subfloat[sensitivity (continuous)]{\includegraphics[width=0.45\linewidth]{continuous/sensitivity_4_2.pdf}}%
%     \subfloat[FDR (continuous)]{\includegraphics[width=0.45\linewidth]{continuous/FDR_4_2.pdf}}
%     \\
%     \subfloat[sensitivity (binomial)]{\includegraphics[width=0.45\linewidth]{binomial/sensitivity_4_2.pdf}}%
%     \subfloat[FDR (binomial)]{\includegraphics[width=0.45\linewidth]{binomial/FDR_4_2.pdf}}
%     \\
%     \subfloat[sensitivity (survival)]{\includegraphics[width=0.45\linewidth]{survival/sensitivity_4_2.pdf}}%
%     \subfloat[FDR (survival)]{\includegraphics[width=0.45\linewidth]{survival/FDR_4_2.pdf}}
%     \caption{Sensitivity and false discovery rate for predictive biomarkers by using different model selection criteria in Model IV for $n=200$.}
%     \label{fig:criteria_sensitivity}
% \end{figure}

We generate a matrix $\bm{X}$ of $n$ by $p$ from the standard normal distribution, and a treatment vector $\bm{t}$ being either $-1$ or $1$, where $1$ represents the treatment group and $-1$ for the control group. The way to define the treatment vector is to allow the variance of the treatment to be 1. For continuous response, the response variable is generated by adding a linear combination of groups of prognostic and predictive terms and the Gaussian noise, such as $\bm{y}=\tau \bm{t} + \sum_{j=1}^{d} \{\beta_j\bm{x}_j+\gamma_j\bm{x}_j \times \bm{t}\} + \bm{\varepsilon}$, where $\bm{\varepsilon}$ follows the Gaussian distribution with mean $\bm{0}$ and variance-covariance matrix $0.2\bm{I}$. The treatment effect $\tau$ is 0.63 for all four models, while the prognostic and predictive effects for each model are set as in Table~\ref{tab:model_desc}. The signal-to-noise ratio (SNR) resulted from prognostic and predictive effects is set 1, respectively.

For binomial scenario, the response variable is generated by the logistic regression such that $\mathrm{P}(y_i=1) =1/(1+\mathrm{exp}(-y_{i}))$, where $y_{i}$ is the $i$th element in $\bm{y}$. For survival case, the response variable is generated based on the Cox proportional hazard model $h(t)=h_0(t)\mathrm{exp}(\bm{y})$ by using the R package \texttt{coxed} \citep{harden2018simulating}, in which the maximum survival time is 100, the censored ratio is 0.1, and the baseline hazard function is generated using the default flexible-hazard method as described in~\cite{harden2018simulating}. 

\subsection{Accuracy, interpretability, and predictability}\label{subsec:sim_sense}

For each scenario, we generate a training and testing data in the same way, where a model is fitted on the training data that is then used to investigate the prediction power for the testing data. In order to investigate the performances of different methods on selecting the correct prognostic and predictive biomarkers from a large number of biomarkers, we calculate the $\mathrm{F_1}$ scores for the prognostic and predictive effects, the hierarchy enforcement rate, and the mean-square errors on the testing data, respectively. 

\begin{figure}
        \centering
        \includegraphics[width=\linewidth,height = 3in]{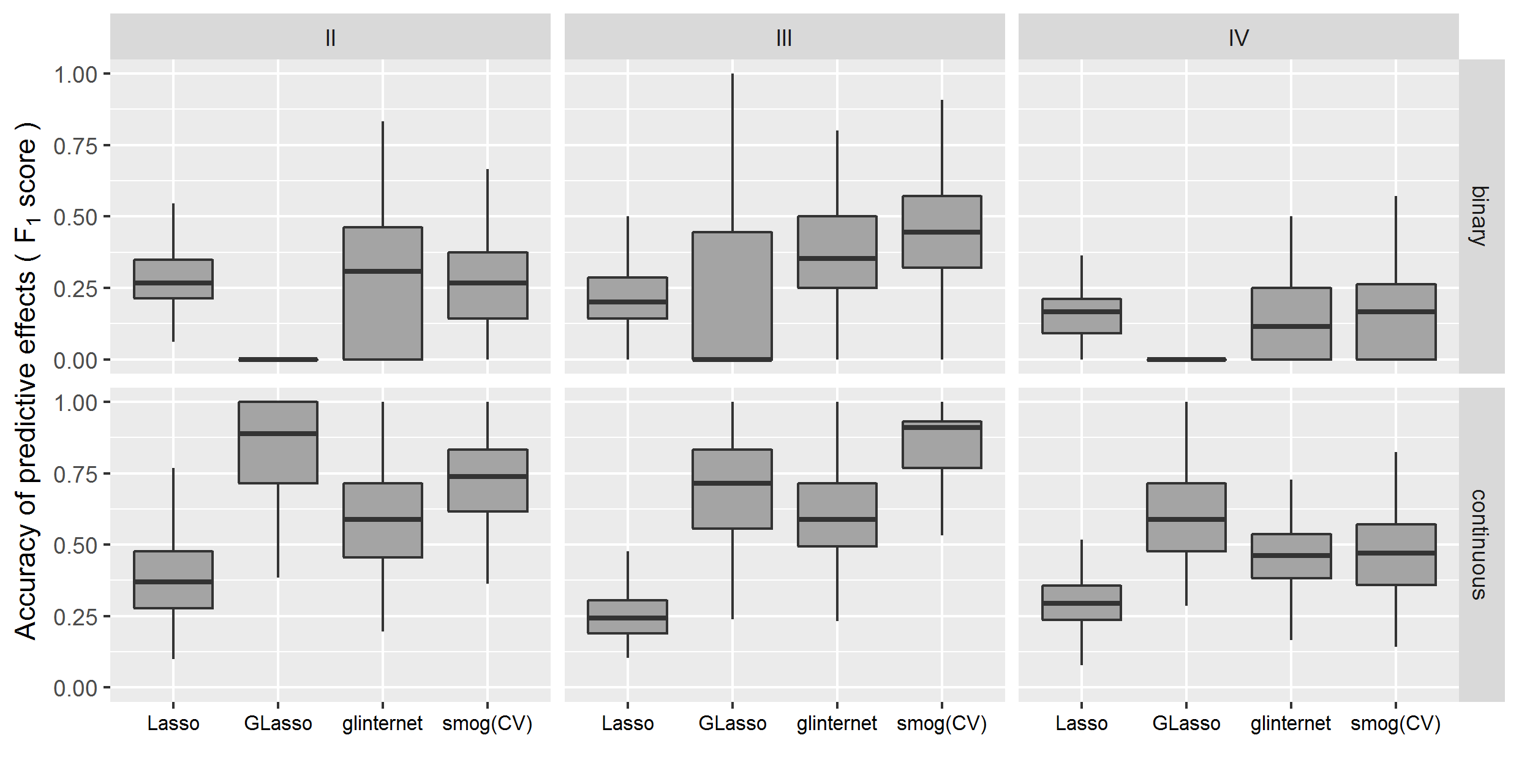}
        \caption{Accuracy of selected predictive effects ($\mathrm{F_1}$ score) for continuous and binary data: n=200, d=1000.}
        \label{fig:pred_accu2}
\end{figure}

\setlength\extrarowheight{4pt}
\begin{table}[hp]
\centering
\caption{\footnotesize{$\mathrm{F}_1$ scores for prognostic and predictive effects, hierarchy enforcement rate, and mean square errors (negative log-likelihood values) for Lasso, group Lasso, Glinternet, and smog using cross-validation model assessments in four scenarios for continuous and binomial data.}}
\label{tab:sim_sum}
% \begin{adjustbox}{max width=\textwidth}
\resizebox{\textwidth}{!}{
\begin{tabular}{@{}C{1.25cm}*{22}{c}@{}}
Continuous:
    &  &  &  &  &  &  &  &  &  &  &  &  &  &  &  &  &  &  &  &  &  & \\
  \hline \hline 
 \multirow{2}{*}{Scenario} & 
 \multirow{2}{*}{n} & 
 \multirow{2}{*}{p} 
 & & \multicolumn{4}{c}{Lasso} 
 & & \multicolumn{4}{c}{GLasso} 
 & & \multicolumn{4}{c}{Glinternet} 
 & & \multicolumn{4}{c}{smog(CV)} \\
\cline{5-8} \cline{10-13} \cline{15-18} \cline{20-23}
 & & & 
& prog & pred & hierarchy & mse & 
& prog & pred & hierarchy & mse & 
& prog & pred & hierarchy & mse & 
& prog & pred & hierarchy & mse \\ 
  \hline
  \multirow{6}{*}{I} 
  & 100 & 200 &  & 
  0.47 & 0.00 & 0.01 & 0.31 &  
  & \textbf{0.86} & 0.00 & 0.54 & \textbf{0.31} &  
  & 0.38 & 0.00 & 1.00 & 0.29 &  
  & \textit{\textbf{0.49}} & 0.00 & 1.00 & \textit{\textbf{0.27}} \\ 
   & 100 & 500 &  
   & 0.40 & 0.00 & 0.02 & 0.35 &  
   & \textbf{0.77} & 0.00 & 0.50 & \textbf{0.33} &  
   & 0.30 & 0.00 & 1.00 & 0.32 &  
   & \textit{\textbf{0.46}} & 0.00 & 1.00 & \textit{\textbf{0.30}} \\ 
   & 100 & 1000 &  
   & 0.36 & 0.00 & 0.01 & 0.37 &  
   & \textbf{0.74} & 0.00 & 0.58 & \textbf{0.36} &  
   & 0.27 & 0.00 & 1.00 & 0.35 &  
   & \textit{\textbf{0.44}} & 0.00 & 1.00 & \textit{\textbf{0.32}} \\ 
   & 200 & 200 &  
   & 0.48 & 0.00 & 0.00 & 0.25 & 
   & \textbf{0.94} & 0.00 & 0.76 & \textbf{\textit{0.25}} &  
   & 0.40 & 0.00 & 1.00 & 0.24 &  
   & \textit{\textbf{0.70}} & 0.00 & 1.00 & \textit{\textbf{0.24}} \\ 
   & 200 & 500 &  
   & 0.42 & 0.00 & 0.00 & 0.26 &  
   & \textbf{0.94} & 0.00 & 0.72 & \textbf{0.26} &  
   & 0.33 & 0.00 & 1.00 & 0.25 &  
   & \textit{\textbf{0.58}} & 0.00 & 1.00 & \textit{\textbf{0.24}} \\ 
   & 200 & 1000 &  
   & 0.35 & 0.00 & 0.00 & 0.27 &  
   & \textbf{0.92} & 0.00 & 0.68 & \textbf{0.27} &  
   & 0.27 & 0.00 & 1.00 & 0.26 &  
   & \textit{\textbf{0.60}} & 0.00 & 1.00 & \textit{\textbf{0.24}} \\ 
   \hline 
   \multirow{6}{*}{II}
   & 100 & 200 &  
   & 0.00 & 0.43 & 0.00 & 0.32 &  
   & 0.00 & \textbf{0.73} & 0.01 & \textbf{0.34} &  
   & 0.00 & 0.53 & 1.00 & 0.35 &  
   & 0.00 & \textbf{0.53} & 1.00 & \textbf{0.36} \\ 
   & 100 & 500 &  
   & 0.00 & 0.41 & 0.00 & 0.33 &  
   & 0.00 & \textbf{0.68} & 0.08 & \textbf{0.37} &  
   & 0.00 & 0.47 & 1.00 & 0.38 &  
   & 0.00 & \textit{\textbf{0.44}} & 1.00 & \textit{\textbf{0.37}} \\ 
   & 100 & 1000 &  
   & 0.00 & 0.35 & 0.00 & 0.37 &  
   & 0.00 & \textbf{0.59} & 0.11 & \textbf{0.40} &  
   & 0.00 & 0.38 & 1.00 & 0.42 &  
   & 0.00 & \textit{\textbf{0.30}} & 1.00 & \textit{\textbf{0.40}} \\ 
   & 200 & 200 &  
   & 0.00 & 0.47 & 0.00 & 0.25 &  
   & 0.00 & \textbf{0.89} & 0.00 & \textbf{0.27} &  
   & 0.00 & 0.57 & 1.00 & 0.26 &  
   & 0.00 & \textit{\textbf{0.68}} & 1.00 & \textit{\textbf{0.27}} \\ 
   & 200 & 500 &  
   & 0.00 & 0.41 & 0.00 & 0.26 &  
   & 0.00 & \textbf{0.87} & 0.00 & \textbf{0.28} &  
   & 0.00 & 0.59 & 1.00 & 0.28 &  
   & 0.00 & \textit{\textbf{0.70}} & 1.00 & \textit{\textbf{0.29}} \\ 
   & 200 & 1000 &  
   & 0.00 & 0.38 & 0.00 & 0.27 & 
   & 0.00 & \textbf{0.85} & 0.00 & \textbf{0.29} &  
   & 0.00 & 0.59 & 1.00 & 0.30 &  
   & 0.00 & \textit{\textbf{0.74}} & 1.00 & \textit{\textbf{0.30}} \\
   \hline 
   \multirow{6}{*}{III}
   & 100 & 200 &  
   & 0.34 & 0.34 & 0.01 & 0.44 &  
   & 0.64 & \textbf{0.61} & 0.22 & \textbf{0.41} & 
   & 0.28 & 0.52 & 1.00 & 0.40 & 
   & 0.55 & \textbf{0.67} & 1.00 & \textbf{0.42} \\ 
   & 100 & 500 &  
   & 0.30 & 0.29 & 0.01 & 0.52 &  
   & 0.61 & \textbf{0.54} & 0.30 & \textbf{0.49} &  
   & 0.25 & 0.46 & 1.00 & 0.48 &  
   & 0.49 & \textit{\textbf{0.57}} & 1.00 & \textit{\textbf{0.47}} \\ 
   & 100 & 1000 &  
   & 0.25 & 0.24 & 0.06 & 0.57 &  
   & 0.55 & \textbf{0.41} & 0.32 & \textbf{0.55} & 
   & 0.20 & 0.40 & 1.00 & 0.54 & 
   & 0.42 & \textit{\textbf{0.47}} & 1.00 & \textit{\textbf{0.51}} \\ 
   & 200 & 200 &  
   & 0.35 & 0.35 & 0.00 & 0.27 &  
   & 0.70 & \textbf{0.74} & 0.22 & \textbf{0.28} & 
   & 0.29 & 0.60 & 1.00 & 0.27 & 
   & 0.74 & \textit{\textbf{0.92}} & 1.00 & \textit{\textbf{0.30}} \\ 
   & 200 & 500 &  
   & 0.30 & 0.30 & 0.00 & 0.30 & 
   & 0.66 & \textbf{0.72} & 0.18 & \textbf{0.30} & 
   & 0.25 & 0.64 & 1.00 & 0.29 & 
   & 0.66 & \textit{\textbf{0.88}} & 1.00 & \textit{\textbf{0.31}} \\ 
   & 200 & 1000 &  
   & 0.25 & 0.26 & 0.00 & 0.33 & 
   & 0.64 & \textbf{0.68} & 0.14 & \textbf{0.31} & 
   & 0.20 & 0.59 & 1.00 & 0.31 & 
   & 0.70 & \textit{\textbf{0.87}} & 1.00 & \textit{\textbf{0.32}} \\ 
   \hline 
   \multirow{6}{*}{IV}
   & 100 & 200 &  
   & 0.34 & 0.33 & 0.01 & 0.56 & 
   & 0.47 & \textbf{0.42} & 0.21 & \textbf{0.56} & 
   & 0.30 & 0.36 & 1.00 & 0.55 & 
   & 0.39 & \textit{\textbf{0.33}} & 1.00 & \textit{\textbf{0.55}} \\ 
   & 100 & 500 & 
   & 0.24 & 0.21 & 0.03 & 0.62 & 
   & 0.36 & \textbf{0.26} & 0.36 & \textbf{0.64} & 
   & 0.22 & 0.23 & 1.00 & 0.62 & 
   & 0.28 & \textit{\textbf{0.22}} & 1.00 & \textit{\textbf{0.59}} \\ 
   & 100 & 1000 &  
   & 0.14 & 0.14 & 0.02 & 0.65 & 
   & 0.28 & \textbf{0.15} & 0.57 & \textbf{0.69} & 
   & 0.14 & 0.13 & 1.00 & 0.64 & 
   & 0.20 & \textit{\textbf{0.15}} & 1.00 & \textit{\textbf{0.60}} \\ 
   & 200 & 200 &  
   & 0.41 & 0.41 & 0.00 & 0.35 & 
   & 0.67 & \textbf{0.68} & 0.00 & \textbf{0.34} &  
   & 0.30 & 0.56 & 1.00 & 0.36 & 
   & 0.58 & \textit{\textbf{0.68}} & 1.00 & \textit{\textbf{0.42}} \\ 
   & 200 & 500 & 
   & 0.35 & 0.35 & 0.00 & 0.43 & 
   & 0.61 & \textbf{0.65} & 0.00 & \textbf{0.39} & 
   & 0.26 & 0.53 & 1.00 & 0.42 & 
   & 0.46 & \textit{\textbf{0.55}} & 1.00 & \textit{\textbf{0.46}} \\ 
   & 200 & 1000 &  
   & 0.30 & 0.30 & 0.00 & 0.48 & 
   & 0.58 & \textbf{0.59} & 0.02 & \textbf{0.42} & 
   & 0.22 & 0.46 & 1.00 & 0.47 & 
   & 0.45 & \textit{\textbf{0.46}} & 1.00 & \textit{\textbf{0.49}}\\ 
   \hline \hline 
   &  &  &  &  &  &  &  &  &  &  &  &  &  &  &  &  &  &  &  &  &  & \\ 
   
    Binomial:
    &  &  &  &  &  &  &  &  &  &  &  &  &  &  &  &  &  &  &  &  &  & \\ 
    \hline \hline 
      \multirow{6}{*}{I} 
 & 100 & 200 &  
 & 0.30 & 0.00 & 0.11 & 35.75 &  
 & \textbf{0.37} & 0.00 & 0.94 & \textbf{55.12} &  
 & 0.27 & 0.00 & 1.00 & 34.38 &  
 & \textit{\textbf{0.33}} & 0.00 & 1.00 & \textit{\textbf{28.35}} \\ 
 & 100 & 500 &  
 & 0.20 & 0.00 & 0.06 & 36.85 &  
 & \textbf{0.36} & 0.00 & 0.96 & \textbf{55.72} &  
 & 0.17 & 0.00 & 1.00 & 35.69 &  
 & \textit{\textbf{0.23}} & 0.00 & 1.00 & \textit{\textbf{28.86}} \\ 
 & 100 & 1000 &  
 & 0.13 & 0.00 & 0.08 & 38.56 &  
 & \textbf{0.35} & 0.00 & 0.96 & \textbf{56.41} &  
 & 0.12 & 0.00 & 1.00 & 37.20 &  
 & \textit{\textbf{0.16}} & 0.00 & 1.00 & \textit{\textbf{30.19}} \\ 
 & 200 & 200 &  
 & 0.42 & 0.00 & 0.02 & 57.43 &  
 & \textbf{0.44} & 0.00 & 0.93 & \textbf{106.98} &  
 & 0.34 & 0.00 & 1.00 & 55.36 &  
 & \textit{\textbf{0.42}} & 0.00 & 1.00 & \textit{\textbf{48.75}} \\ 
 & 200 & 500 &  
 & 0.34 & 0.00 & 0.03 & 61.43 &  
 & \textbf{0.42} & 0.00 & 0.94 & \textbf{109.84} &  
 & 0.26 & 0.00 & 1.00 & 59.21 &  
 & \textit{\textbf{0.35}} & 0.00 & 1.00 & \textit{\textbf{51.94}} \\ 
 & 200 & 1000 &  
 & 0.28 & 0.00 & 0.06 & 63.46 &  
 & \textbf{0.37} & 0.00 & 0.92 & \textbf{112.47} &  
 & 0.23 & 0.00 & 1.00 & 61.40 &  
 & \textit{\textbf{0.31}} & 0.00 & 1.00 & \textit{\textbf{53.46}} \\ 
   \hline 
    \multirow{6}{*}{II} 
   & 100 & 200 &  
   & 0.00 & \textbf{0.27} & 0.02 & \textbf{36.03} &  
   & 0.00 & 0.03 & 0.95 & 57.20 &  
   & 0.00 & 0.21 & 1.00 & 36.71 &  
   & 0.00 & \textit{\textbf{0.20}} & 1.00 & \textit{\textbf{30.92}} \\ 
   & 100 & 500 &  
   & 0.00 & \textbf{0.19} & 0.03 & \textbf{37.32} &  
   & 0.00 & 0.03 & 0.94 & 57.56 &  
   & 0.00 & 0.13 & 1.00 & 38.08 &  
   & 0.00 & \textit{\textbf{0.12}} & 1.00 & \textit{\textbf{31.01}} \\ 
   & 100 & 1000 &  
   & 0.00 & \textbf{0.14} & 0.03 & \textbf{38.66} &  
   & 0.00 & 0.03 & 0.94 & 57.08 &  
   & 0.00 & 0.08 & 1.00 & 39.30 &  
   & 0.00 & \textit{\textbf{0.08}} & 1.00 & \textit{\textbf{30.86}} \\ 
   & 200 & 200 &  
   & 0.00 & 0.41 & 0.00 & 57.62 &  
   & 0.00 & 0.09 & 0.88 & 114.58 &  
   & 0.00 & \textbf{0.46} & 1.00 & \textbf{60.94} &  
   & 0.00 & \textit{\textbf{0.45}} & 1.00 & \textit{\textbf{55.93}} \\ 
   & 200 & 500 &  
   & 0.00 & 0.36 & 0.01 & 61.02 &  
   & 0.00 & 0.05 & 0.94 & 117.10 &  
   & 0.00 & \textbf{0.38} & 1.00 & \textbf{65.67} &  
   & 0.00 & \textit{\textbf{0.33}} & 1.00 & \textit{\textbf{58.32}} \\ 
   & 200 & 1000 &  
   & 0.00 & 0.28 & 0.02 & 64.15 &  
   & 0.00 & 0.04 & 0.94 & 117.31 &  
   & 0.00 & \textbf{0.29} & 1.00 & \textbf{68.33} &  
   & 0.00 & \textit{\textbf{0.25}} & 1.00 & \textit{\textbf{59.52}} \\ 
   \hline 
    \multirow{6}{*}{III} 
   & 100 & 200 &  
   & 0.21 & 0.21 & 0.06 & 36.04 &  
   & 0.40 & 0.11 & 0.78 & 45.79 &  
   & 0.25 & \textbf{0.27} & 1.00 & \textbf{35.18} &  
   & 0.33 & \textit{\textbf{0.34}} & 1.00 & \textit{\textbf{35.44}} \\ 
   & 100 & 500 &  
   & 0.14 & 0.14 & 0.05 & 38.38 &  
   & 0.37 & 0.09 & 0.74 & 47.10 &  
   & 0.15 & \textbf{0.15} & 1.00 & \textbf{37.94} &  
   & 0.22 & \textit{\textbf{0.24}} & 1.00 & \textit{\textbf{36.98}} \\ 
   & 100 & 1000 &  
   & 0.08 & 0.09 & 0.07 & 40.42 &  
   & 0.36 & 0.06 & 0.76 & 47.85 &  
   & 0.09 & \textbf{0.10} & 1.00 & \textbf{39.82} &  
   & 0.15 & \textit{\textbf{0.16}} & 1.00 & \textit{\textbf{37.92}} \\ 
   & 200 & 200 &  
   & 0.34 & 0.35 & 0.01 & 56.97 &  
   & 0.56 & 0.35 & 0.54 & 78.57 &  
   & 0.32 & \textbf{0.47} & 1.00 & \textbf{54.85} &  
   & 0.48 & \textit{\textbf{0.58}} & 1.00 & \textit{\textbf{59.41}} \\ 
   & 200 & 500 &  
   & 0.27 & 0.27 & 0.01 & 62.74 &  
   & 0.52 & 0.25 & 0.66 & 83.36 &  
   & 0.27 & \textbf{0.43} & 1.00 & \textbf{60.79} &  
   & 0.42 & \textit{\textbf{0.51}} & 1.00 & \textit{\textbf{64.47}} \\ 
   & 200 & 1000 &  
   & 0.20 & 0.22 & 0.01 & 65.00 &  
   & 0.44 & 0.20 & 0.70 & 85.92 &  
   & 0.21 & \textbf{0.36} & 1.00 & \textbf{63.77} &  
   & 0.36 & \textit{\textbf{0.43}} & 1.00 & \textit{\textbf{67.76}} \\ 
   \hline 
    \multirow{6}{*}{IV} 
   & 100 & 200 &  
   & 0.16 & \textbf{0.15} & 0.06 & \textbf{47.04} &  
   & 0.21 & 0.03 & 0.85 & 56.46 &  
   & 0.17 & 0.11 & 1.00 & 46.59 &  
   & 0.22 & \textit{\textbf{0.12}} & 1.00 & \textit{\textbf{42.28}} \\ 
   & 100 & 500 &  
   & 0.09 & \textbf{0.08} & 0.06 & \textbf{49.42} &  
   & 0.20 & 0.02 & 0.83 & 57.33 &  
   & 0.09 & 0.04 & 1.00 & 48.48 &  
   & 0.12 & \textit{\textbf{0.08}} & 1.00 & \textit{\textbf{43.09}} \\ 
   & 100 & 1000 &  
   & 0.05 & \textbf{0.05} & 0.07 & \textbf{50.18} &  
   & 0.18 & 0.02 & 0.89 & 57.57 &  
   & 0.05 & 0.03 & 1.00 & 49.60 &  
   & 0.08 & \textit{\textbf{0.05}} & 1.00 & \textit{\textbf{43.68}} \\ 
   & 200 & 200 &  
   & 0.31 & 0.32 & 0.01 & 80.78 &  
   & 0.24 & 0.09 & 0.83 & 110.55 &  
   & 0.31 & \textbf{0.33} & 1.00 & \textbf{79.71} &  
   & 0.37 & \textit{\textbf{0.28}} & 1.00 & \textit{\textbf{78.35}} \\ 
   & 200 & 500 &  
   & 0.18 & 0.21 & 0.04 & 86.18 &  
   & 0.21 & 0.06 & 0.88 & 113.42 &  
   & 0.19 & \textbf{0.21} & 1.00 & \textbf{85.53} &  
   & 0.26 & \textit{\textbf{0.23}} & 1.00 & \textit{\textbf{80.72}} \\ 
   & 200 & 1000 &  
   & 0.15 & \textbf{0.16} & 0.06 & \textbf{89.52} &  
   & 0.20 & 0.03 & 0.91 & 113.44 &  
   & 0.15 & 0.13 & 1.00 & 88.62 &  
   & 0.21 & \textit{\textbf{0.17}} & 1.00 & \textit{\textbf{83.02}} \\ 
    \hline \hline 
\end{tabular}
}
% \end{adjustbox}
\end{table}

$\mathrm{F_1}$ score is the harmonic average of the sensitivity and the positive predictive value, where the sensitivity refers to be the proportion of the correct detected prognostic (predictive) effects to the total number of true prognostic (predictive) effects, and the positive predictive value refers to be the proportion of the correct detected prognostic (predictive) effects to the total number of claimed prognostic (predictive) effects, respectively. $\mathrm{F_1}$ score reaches 1 for the perfect case and 0 for the worst. The hierarchy enforcement rate is the proportion of honoring the hierarchy structure~\eqref{hierstr} over all the repetitions of the simulation, which attains 1 for strictly enforcing the hierarchy, and 0 for fully disobeying. We accordingly repeat 200 times of the simulation for each of the four scenarios in section~\ref{subsec:sim_scen} for continuous, binary, and time-to-event data. 

In Section~\ref{sec:tuning}, we proposed several model assessments and a greedy path search for the optimal penalty parameters. In the supplemental material, we investigate the performance of these model assessments using our proposed method \texttt{smog}, such as the mean squared error of the response in five-fold cross-validations (CV), the correction AIC (cAIC), AIC, BIC, and the generalized cross-validation score (GCV), respectively. Comprehensively speaking, these model assessments have very comparable performances with respect to the $\mathrm{F_1}$ scores of prognostic and predictive effects, the hierarchy enforcement rates, and the predictive power. Obviously, the cross-validation criterion is more computationally expensive than other model assessments.

Table~\ref{tab:sim_sum} compare our proposed method using cross-validation model assessment with the competitor methods, such as \texttt{Lasso}, group-Lasso without honoring hierarchy (\texttt{GLasso}), and group-Lasso with honoring strong hierarchy (\texttt{Glinternet}) with respect to the $\mathrm{F_1}$ scores of prognostic and predictive effect, the hierarchy enforcement rate, and the mean squared-error of the testing data, respectively. Figure~\ref{fig:pred_accu2} displays the accuracy of selected predictive biomarkers in $\mathrm{F_1}$ score for continuous and binary data. Table~\ref{tab:sim_sum} summarizes for continuous and binary data, and the relevant simulation results for time-to-event data can be found in the supplementary file, in which we merely compare our proposed method with \texttt{Lasso} because neither \texttt{GLasso} nor \texttt{Glinternet} can deal with the time-to-event data. 

Thoroughly speaking, the simulation results demonstrate that our method \texttt{smog} has very competitive, if not the best, predictive performance than the competitors in all scenarios for different types of continuous, binary, and time-to-event data, and \texttt{smog} always strictly obey the hierarchy structures~\eqref{hierstr}. Most importantly, \texttt{smog} dominates, if not overwhelmingly, on the accuracy of detecting prognostic and predictive effects in most cases, especially, when the strong hierarchy exists in biomarkers in scenarios III and IV. 

On the other hand, in scenario I, when the weak hierarchy lies in biomarkers, \texttt{GLasso} has somewhat bigger $\mathrm{F_1}$ scores for prognostic effects than \texttt{smog}, though it never strictly obey the hierarchy enforcement. Nonetheless, because \texttt{GLasso} would select the prognostic and predictive effects simultaneously, higher $\mathrm{F}_1$ score on prognostic effect would select more false predictive biomarkers, while $\mathrm{F_1}$ score for predictive effect would be always zero whatever the number of false predictive biomarkers in that scenario I does not have any predictive biomarker. In scenario II, \texttt{smog} has better or commensurate performance on the accuracy of detecting predictive effects than \texttt{Glinternet}, while \texttt{smog} has somewhat smaller $\mathrm{F_1}$ scores than \texttt{GLasso}.

\section{Biomarker study and subgroup analysis for Head and neck cancer}\label{sec:Data analysis}
\subsection{Head and neck cancer}\label{subsec:hnscc}
Head and neck squamous cell carcinoma (HNSCC) is the seventh most frequent cancer worldwide and accounts for about 4\% of all cancers in the United States. Generally, radiation therapy is a useful treatment for the patients with the HNSCC. However, recent research finds that the HNSCC patient population may have different response to the radiation treatment based on their genetic characteristics. For example, the patients with HPV-positive oropharyngeal tumors have a better prognosis than those with those with HPV-negative oropharyngeal tumors. Hence, it is crucial to identify the genetic subpopulation of HNSCC patients with the radiation treatment effect for the radiation therapy. 

\subsection{Data Preprocess}\label{subsec:hnscc_preproc}
We screened and collected 510 HNSCC patients clinical data from the TCGA and cBioPortal websites, which included 301 patients with radiation therapy and 209 patients without radiation therapy. The data contained the survival information and the whole genomic profile that consist of $20,531$ gene expressions for each patient. In order to better implement our proposed approach, we preprocessed the gene expression data by removing the genes with missing values, and merely keeping those genes with the coefficient of variation (CV) between 0.7 and 10. Then, we took the natural log-transformation on the remaining $3,973$ genes and then normalized those gene expressions. To the end, we fed these normalized $3,973$ genes into our algorithm. For the sake of the external validation on the performance of our method, we randomly split the whole data into a learning data set and a validation data set with equal sample size. 

\subsection{Biomarker study and subgroup analysis}\label{subsec:hnscc_sum}
Fitting on the learning data using our method, we obtained a sparse model that contained 25 predictive genes with both prognostic and predictive effects, and 14 prognostic genes with merely prognostic effects. Figure~\ref{fig:HNSC_opt} illustrates the greedy pathway search for the optimal $\lambda_1$ and $\lambda_3$ based on the model assessment GCV, where the minimum GCV is $0.0392$ that is obtained at $\lambda_1^{opt}=0.174$ and $\lambda_3^{opt}=0.215$. $\lambda_2$ is kept being zero by default. And the fitted model honors the hierarchy structure~\eqref{hierstr} such that an existing predictive effect must enforce the corresponding prognostic effect in the model, and the estimated treatment effect is $-0.001$, which should not be penalized during the model fitting. The details can be seen in the supplementary file. 

Notice that most of the 25 predictive genes have negative predictive effects, which implies that over-expressed these genes could benefit from the radiation therapy, except SRPX, SPOCK1, and one non-coding gene LOC115165. A lot of selected genes like KRT13, CDKN2A, EPHX3, PRSS12 are already discovered to regulate the progression of the head and neck squamous carcinoma cells. 
Recent research finds that KRT13 is repressed in oral squamous cell carcinoma, which may imply over-expressed KRT13 could alleviate the head and neck cancer \citep{naganuma2014epigenetic}. CDKN2A is a tumor-suppressor gene in head and neck squamous cell carcinoma, which can make p16(INK4A) and p14(ARF) to control the cell division, and further stimulate or block cell cycle progression \citep{ai2003p16}. 

The magnitude of the coefficients in the fitted model is quite small that may be penalized too much. We refit the training data on the selected 39 genes by using small penalty parameters such as $\lambda_1=\lambda_3=10^{-6}$, which sheds similar light on the relaxed Lasso \citep{meinshausen2007relaxed}. The prognostic and predictive effects for the selected 39 genes are all nonzeros in the refitted model. In order to investigate the predictive power of the fitted model, we calculate the individual treatment contrast (global predictive effect) by multiplying the selected predictive gene expressions by their corresponding predictive effects, termed $z_i$ for each patient $i=1,\ldots,n$. 
We dichotomize all patients into the biomarker-positive and biomarker-negative groups, where the biomarker-positive group is $\{i: z_i < 0\}$, and the biomarker-negative group is $\{i: z_i > 0\}$, respectively.

Figure~\ref{fig:hnsc_surv} demonstrates that the Kaplan-Meier curves for biomarker-positive subgroup and biomarker-negative subgroup for the training data (first row) and the testing data (second row), respectively. In the training data, the point estimates and 95\% confidence intervals of the relative risk for the positive-biomarker patients and negative-biomarker patients are $0.18$ and $(0.09,0.34)$, and $2.29$ and $(1.23,4.28)$, respectively; in the testing data, the point estimates and 95\% confidence intervals of the relative risk for the positive-biomarker patients and negative-biomarker patients are $0.58$ and $(0.34,0.99)$, and $0.87$ and $(0.51,1.51)$, respectively. There might exist overfitting issues in the training data somewhat, however, whatever in the training and testing data, the biomarker-positive patients benefit more from the radiation treatment, yet the biomarker-negative patients benefit from the control. Collectively, this real case study for HNSCC demonstrates that our method is very competitive to detect very powerful predictive biomarkers for cancer in practice.

\begin{figure}
    \centering
    \includegraphics[width=0.9\textwidth,height=3in]{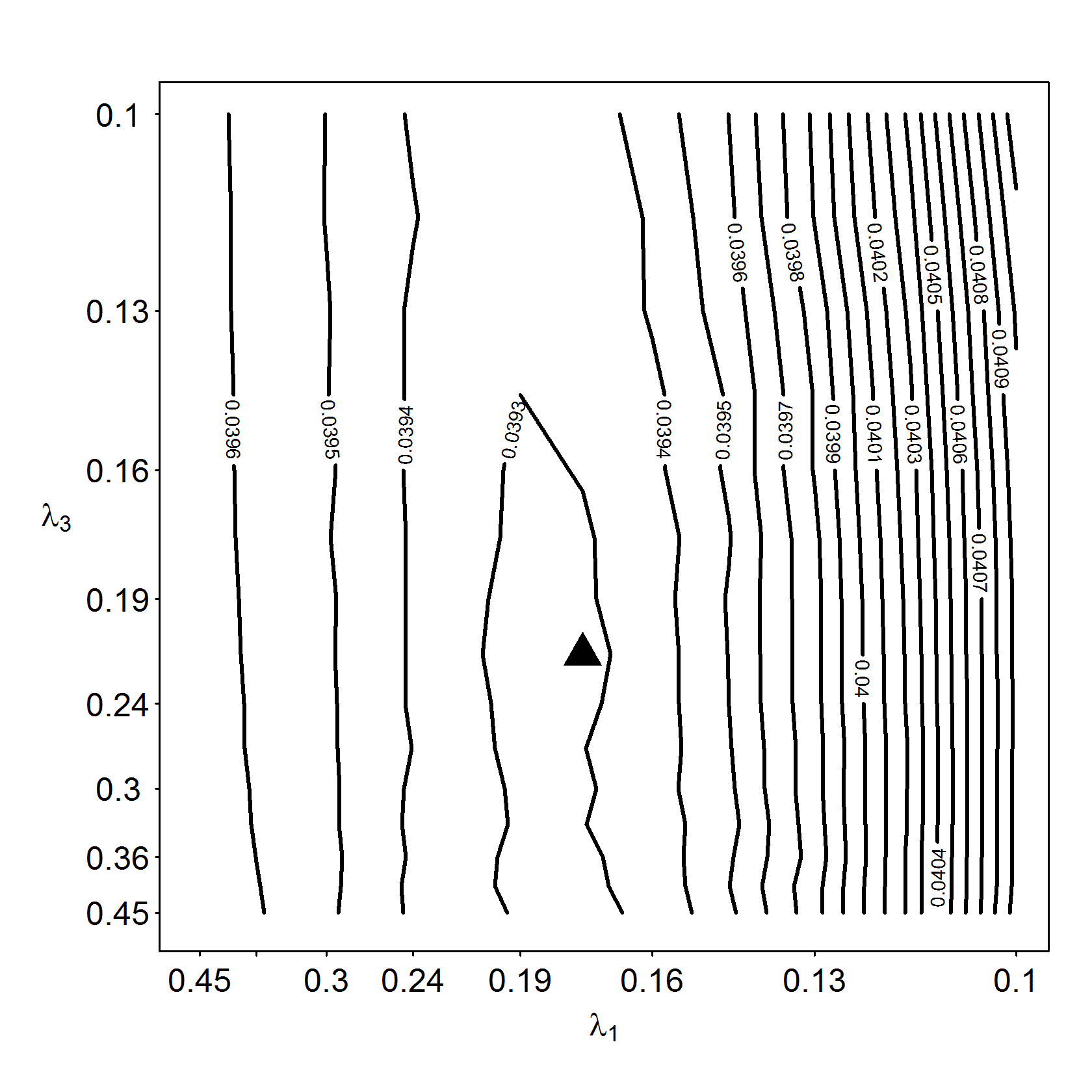}
    \caption{The greedy pathway search for $\lambda_1$ and $\lambda_3$ for modeling the head and neck cancer data based on the model assessment GCV. The minimum GCV locates at the dark triangle sign, where $\lambda_1^{opt}=0.174$ and $\lambda_2^{opt}=0.215$, respectively. }
    \label{fig:HNSC_opt}
\end{figure}

\begin{figure}
    \centering
    \includegraphics[width=0.5\textwidth,height=2in]{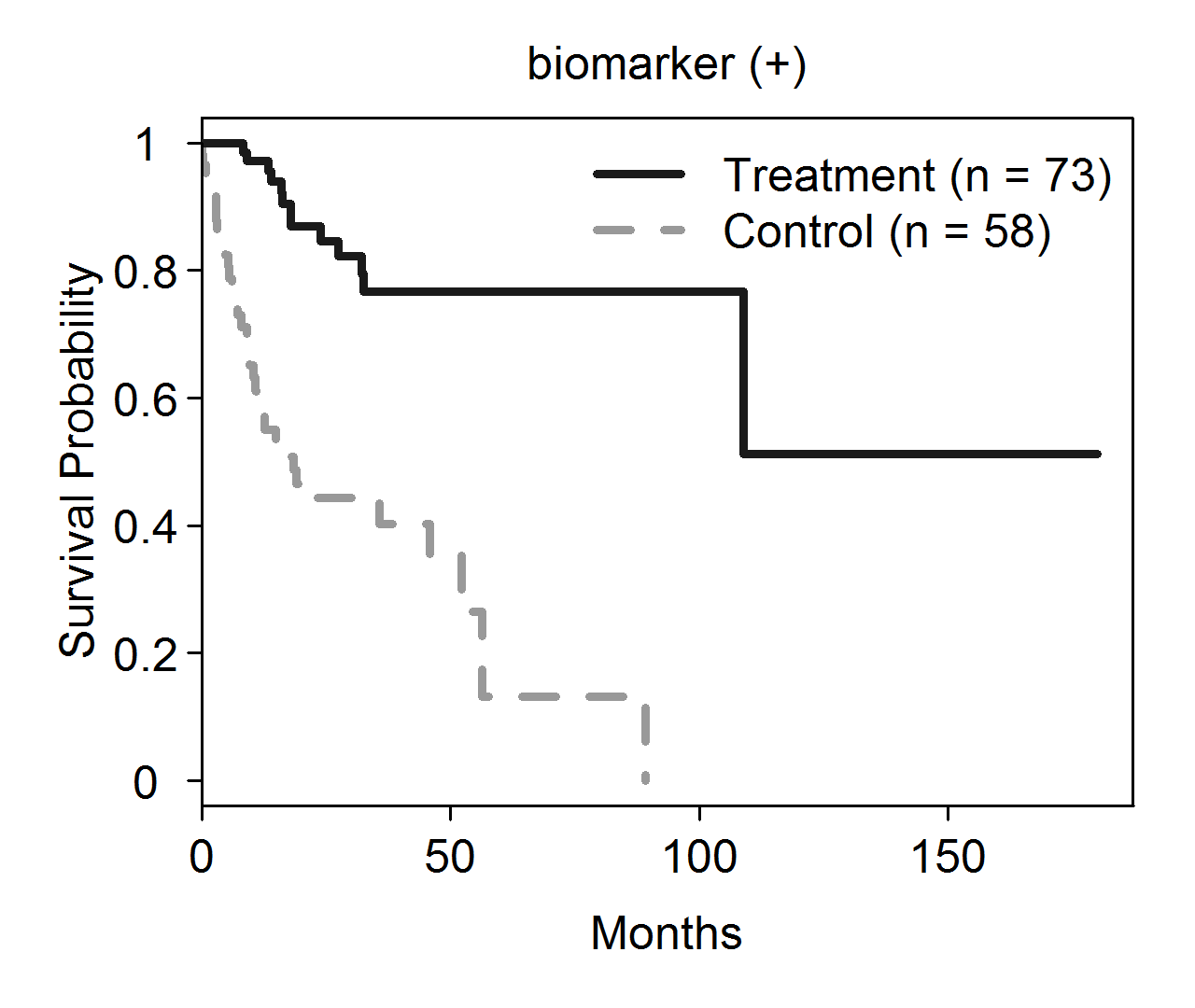}%
    \includegraphics[width=0.5\textwidth,height=2in]{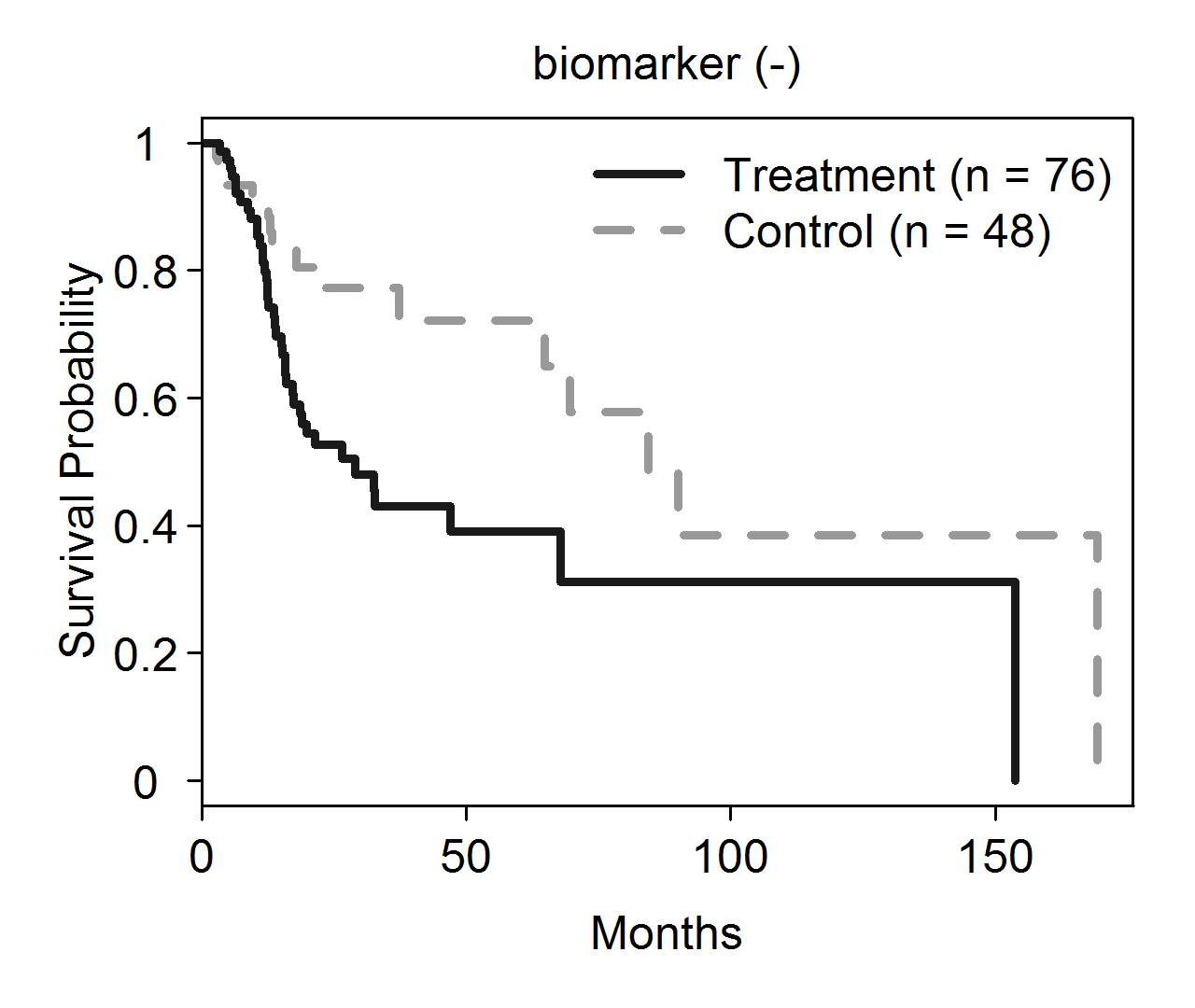}\\
    \includegraphics[width=0.5\textwidth,height=2in]{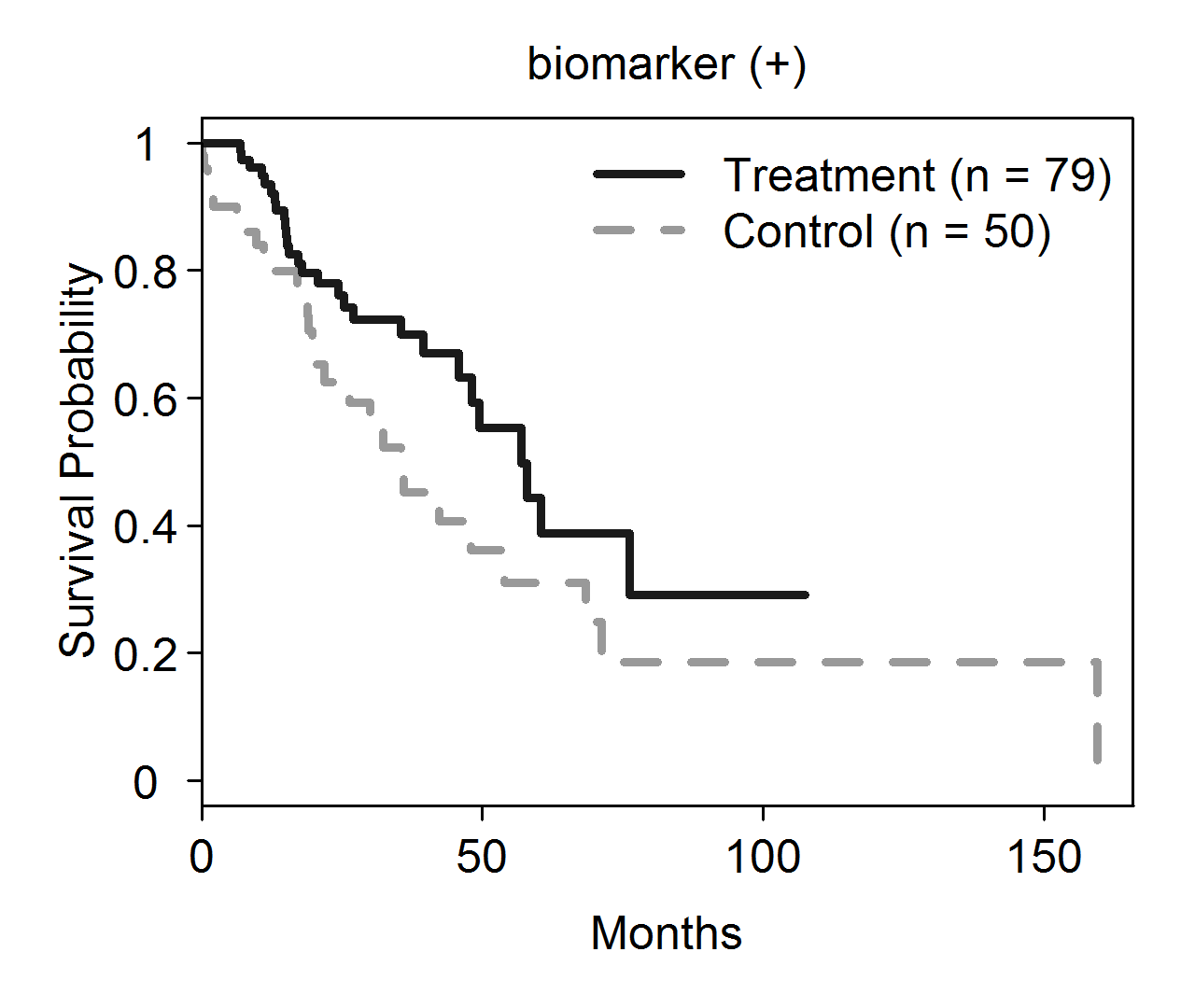}%
    \includegraphics[width=0.5\textwidth,height=2in]{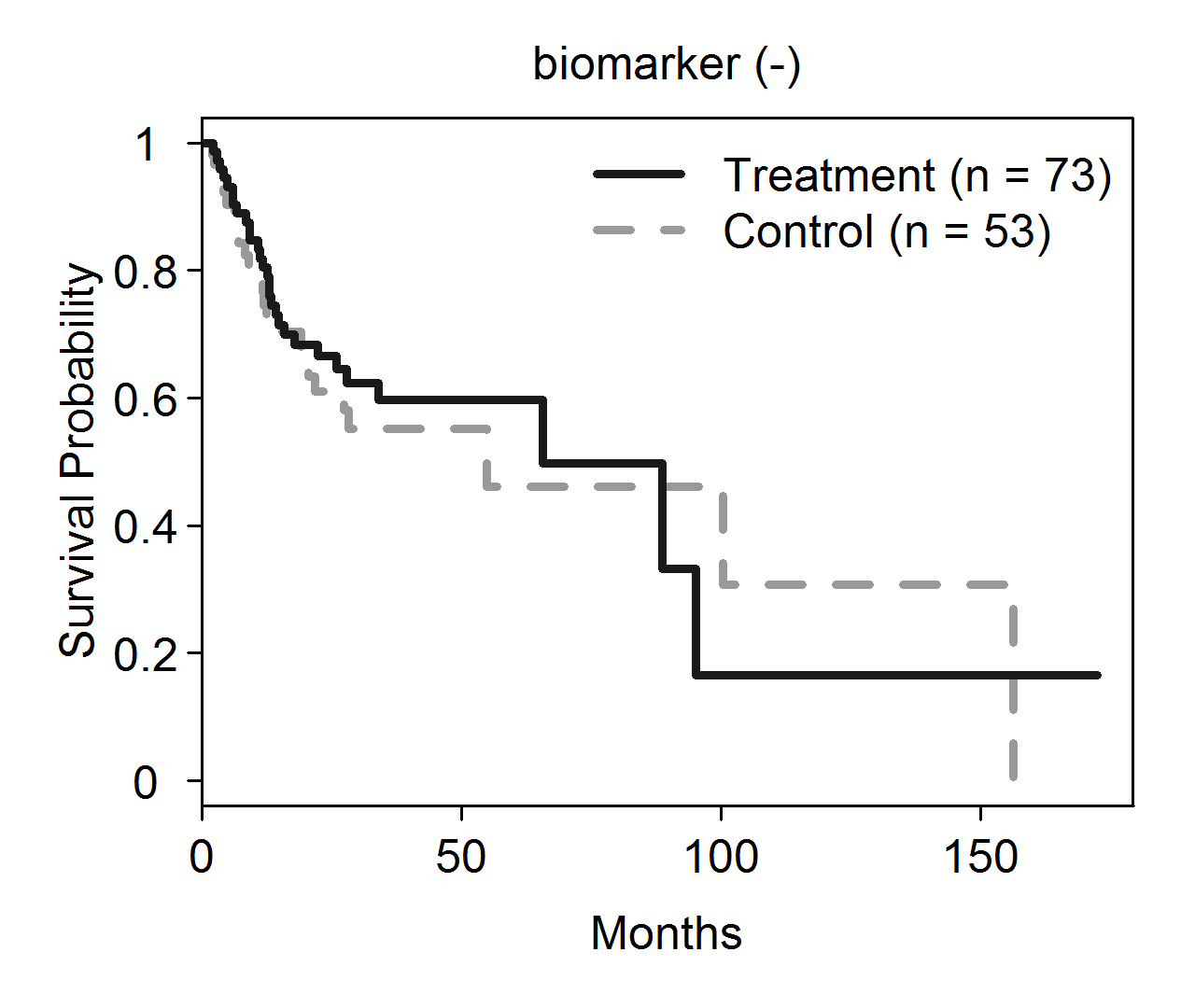}\\
    \caption{Kaplan-Meier curves for the training (first row) and testing (second row) HNSCC data. The first row displays the Kaplan-Meier curves for biomarker-positive and biomarker-negative groups for the training data, where the 95\% confidence intervals of the relative risk for the radiation treatment versus no radiation treatment are $(0.09,0.34)$ and $(1.23,4.28)$, respectively; the second row represents for the biomarker-positive patients and the biomarker-negative patients, where the 95\% confidence intervals of the relative risks for the radiation treatment versus no radiation treatment are $(0.34,0.99)$ and $(0.51,1.51)$, respectively.}
    \label{fig:hnsc_surv}
\end{figure}

% \begin{figure}
%     \centering
%     \includegraphics[width=\textwidth,height=2in]{HNSC_plot/surv_pdg1.pdf}\\
%     \includegraphics[width=\textwidth,height=2in]{HNSC_plot/surv_pdg7.pdf}\\
%     \includegraphics[width=\textwidth,height=2in]{HNSC_plot/surv_pdg10.pdf}\\
%     \includegraphics[width=\textwidth,height=2in]{HNSC_plot/surv_pdg11.pdf}\\
%     \caption{The head and neck cancer patients with over-expressed genes such as ALDH3A1, HLF, and LRRC16B would have higher survival probability by receiving the radiation treatment. However, the survival probability for patients with under-expressed gene LYNX1 would be higher by taking the radiation treatment.}
%     \label{fig:hnsc_surv1}
% \end{figure}

% \begin{figure}
%     \centering
%     \includegraphics[width=\textwidth,height=2in]{HNSC_plot/surv_pdg12.pdf}\\
%     \includegraphics[width=\textwidth,height=2in]{HNSC_plot/surv_pdg14.pdf}\\
%     \includegraphics[width=\textwidth,height=2in]{HNSC_plot/surv_pdg19.pdf}\\
%     \includegraphics[width=\textwidth,height=2in]{HNSC_plot/surv_pdg20.pdf}\\
%     \caption{The head and neck cancer patients with over-expressed genes such as MOXD1, NTRK2, SPP1, and ST6GALNAC1 would have higher survival probability by receiving the radiation treatment.}
%     \label{fig:hnsc_surv2}
% \end{figure}

\section{Discussion} \label{sec:Discuss}
In this article, we propose a novel penalized regression model with a specific penalty function for enforcing the hierarchy structure between the predictive and prognostic effects such that a nonzero predictive effect must induce its ancestor prognostic effect nonzero in the model. For obtaining the minimizer of the objective function, we establish an integrative optimization algorithm by blending the the Majorization-Minimization and the alternating direction method of multipliers algorithm. We present explicitly the generalized algorithm for different types of response variables such as continuous, categorical, and time-to-event data, underlying the Gaussian linear model, logistic regression model, and the Cox proportional hazard model, respectively. We prove the asymptotic property of the proposed method under the Gaussian linear model, which could be extended to the generalized linear model in the exponential family. We build a R package \texttt{smog} to realize the proposed method, and conduct enriched simulation study to demonstrate that our method equips the superior power, compared to the competitors including the Lasso by using the R package \texttt{glmnet}, the group-Lasso without enforcing the hierarchy structure by using the R package \texttt{gglasso}, and the group-Lasso with enforcing the hierarchy structure by using the R package \texttt{glinternet}, etc. We conduct a real case study on head and neck squamous cell carcinoma (HNSCC) by using our method for detecting the predictive genetic biomarkers and subgroup analysis, and the analysis result shows that our method is very promosing in pragmatic. In the future research, we are interested in developing an integrative joint modeling method in biomarker study for discovering the common predictive genome biomarkers for different cancer. 

% \section{Conclusion} \label{sec:Conclude}

% \nocite{*}
\bibliographystyle{apalike}
\bibliography{refer2}

\end{document}